\documentclass[11pt]{article}

\usepackage[T1]{fontenc}
\usepackage{mathpazo}
\usepackage[margin=1in]{geometry}
\usepackage{amsmath,amssymb,amsthm,mathtools}
\usepackage{booktabs,tabularx,array}
\usepackage{enumitem}
\usepackage{microtype}
\usepackage[numbers,sort&compress]{natbib}
\usepackage[dvipsnames]{xcolor}
\usepackage[pdfusetitle,colorlinks=true]{hyperref}
\usepackage{thm-restate}
\usepackage{algorithm,algorithmic}
\newtheorem{theorem}{Theorem}[section]

\newtheorem{proposition}[theorem]{Proposition}
\newtheorem{corollary}[theorem]{Corollary}
\newtheorem{definition}[theorem]{Definition}

\newtheorem{observation}[theorem]{Observation}

\newcommand{\X}{\mathcal X}
\newcommand{\R}{\mathcal R}
\newcommand{\LS}{\operatorname{LS}}
\newcommand{\GS}{\operatorname{GS}}
\newcommand{\Lap}{\operatorname{Lap}}

\newcommand{\e}{\mathrm e}
\newcommand{\ind}{\mathbf 1}
\newcommand{\Prb}{\Pr}

\newcommand{\eps}{\varepsilon}

\newif\ifdraft
\ifdraft
\newcommand\dungnote[1]{\textcolor{brown}{***D\~ung: #1}}
\newcommand{\anil}[1]{\textcolor{blue}{#1}}

\else
\newcommand{\dungnote}[1]{}
\newcommand{\anil}[1]{}

\fi

\DeclareMathOperator*{\argmax}{arg\,max}

\title{\Large \bf Local Sensitivity in Exponential Selection:\\
Failure Modes and Valid Calibrations}

\author{D\~ung Nguyen\thanks{University of Texas at San Antonio. Correspondence to: dungtien.nguyen@utsa.edu}, Anil Vullikanti\thanks{University of Virginia}}

\date{}

\begin{document}

\maketitle

\begin{abstract}

Selection is a task that chooses one element from a finite public candidate range to maximize a data-dependent score.
In differential privacy (DP), the exponential mechanism (EM) samples a candidate at a temperature calibrated to the global sensitivity.
In this paper, we study when dataset-dependent sensitivity can safely replace global sensitivity in private selection.
Naive applications of local sensitivity and smooth sensitivity fail.
Local sensitivities of two neighboring datasets can differ arbitrarily, making their output distributions differ significantly even when they have the same score vectors.
We also prove that no smooth sensitivity calibrated by any fixed smoothness and fixed temperature coefficient can guarantee range-independent $(\eps,\delta)$-DP.


We propose three valid approaches.
First, a private, high-probability upper bound on local sensitivity yields approximate DP, and the method extends to finite higher-order sensitivity hierarchies.
Second, our Propose-Test-Release (PTR) variant privately searches a finite public grid for a temperature scale rather than fixing it in advance.
Third, smooth sensitivity supports several designs.
A candidate-independent smooth geometric construction produces a sensitivity envelope that is admissible under the local dampening framework, which privacy is guaranteed for any admissible envelope. 
Additionally, a separate logarithmic transformation utilizes smooth sensitivity to produce a smoothed candidate score function with advantages: having controlled global sensitivity, and preserving the maximizers of the original utility score, i.e., candidates maximizing the utility.
Both of the designs yield range-independent pure DP.
Besides that, we also give two approximate DP private selectors using smooth sensitivity: a direct EM with smooth sensitivity calibrated to the candidate range and privacy parameters that matches a theoretical lower bound up to some constant factor, and one using a privatized smooth upper scale by analyzing the logarithmic transform of the smoothness.
For every proposed mechanism, we derive a high-probability regret bound under its stated conditions.
\end{abstract}

\newpage
\tableofcontents

\newpage

\section{Introduction}
\label{sec:intro}

Data analysis queries often fall into two types: numeric and selection.
Numeric queries calculate numeric values, for example, a scalar or a vector, directly from the input dataset.
Selection queries take a dataset and a set of candidate answers as input, and select one candidate to optimize some quality score function, i.e., choosing the best candidate.

Differential Privacy (DP)~\cite{DworkRoth2014} is a privacy standard (informally, together with its framework) that protects the privacy of private datasets, by guaranteeing the output distributions on every pair of neighboring datasets to be close.
For numeric queries, DP is usually achieved by injecting noise, drawn from some random distributions, into the true output~\cite{DworkRoth2014}.
Noise is carefully calibrated by sensitivity, a measurement that quantifies the change in output when the input changes slightly.
Global sensitivity is the worst case change, that is publicly known for each query within a universe of datasets, i.e., the set of all valid datasets.
Calibrating noise via finite global sensitivity is straightforward.
For standard private numeric mechanisms, all datasets receive noise from the same law and scale, with the noise magnitude being proportional to the global sensitivity of the queries~\cite{DworkRoth2014,Vadhan2017}.

Local sensitivity and its variants~\cite{Vadhan2017,KarwaEtAl2011}, such as smooth sensitivity~\cite{NRS2007}, are methods to use the geometry around one specific dataset, instead of assuming the worst case as in global sensitivity.
Local sensitivity is the change in the output of a query when a specific input dataset changes a little bit.
While local sensitivity can change arbitrarily across nearby datasets, the smooth sensitivity is an upper bound on local sensitivity, and its values on nearby datasets are not too different from each other, where the relative change is determined by the ``smoothness''.
Noise is then calibrated specifically for each dataset, which allows the output to receive less noise when the private input (and its local neighborhood datasets) has low local sensitivity, improving the accuracy.
The trade-off often lies in the computational cost~\cite{KarwaEtAl2011,NRS2007,NguyenEtAl2023}.

For selection queries, given a private dataset $x$, a public candidate range (set) of answers $\R$ of size $|\R| = m$, a utility function $q(x,r), r\in \R$ indicates how good a candidate $r$ is as the final answer.
Without privacy, it is obvious that an optimal solution is always to return the candidates that maximize the utility function, i.e., maximizers or $\argmax_{r\in \R} q(x,r)$.
The exponential mechanism (EM)~\cite{McSherryTalwar2007} is the standard method to release a private selection.
It assigns more probability mass to candidates with higher utility score, with each candidate $r$'s probability proportional partially to an exponent of $(q(x,r))$ and partially to an exponent of the inverse of a term called temperature, i.e. $\propto \exp(\frac{(q(x,r))}{\operatorname{temp}})$, hence the name "exponential".
In the standard form of EM, the temperature ``$\operatorname{temp}$'' scales with the global sensitivity of $q(.,.)$ and the inverse of privacy parameter $\eps$.
The \emph{regret} of an answer $a$, $\max_{r} q(x,r)-q(x,a) \geq 0$, measures how much utility cost we pay to output $a$ instead of the optimal answer.
High global sensitivity yields high temperature, which spreads the probability mass among candidates, widening the regret bound, while small global sensitivity gives low temperature, which focuses the probability mass towards candidates with high utility score.

While local sensitivity and its variants have been extensively applied to numeric queries for more accurate private mechanisms, they are still underutilized for selection queries.
It is tempting to replace the global sensitivity in EM by these local variants, with the hope that we can tailor the output distribution for each specific dataset when it does not need the highest temperature dictated by the worst case global sensitivity.
In this paper, we study the question: 
\emph{For finite-range private selection, how do we calibrate EM using local sensitivity and its related concepts and what are their rigorous utility bounds?}
The generalized EM~\cite{raskhodnikova2016lipschitz} already improves selection when candidate public sensitivities are available. 
Our problem is different.
We study a different setting: a common data-dependent calibration for all candidates, rather than per-candidate data-independent scales.
We explore three main directions in search of alternatives to global sensitivity:

\begin{enumerate}
    \item A data-dependent local sensitivity can be used for direct temperature scaling after being privatized and one-sided upper bounded on its true value (with high probability).
    Although a one-level local sensitivity can be fairly clean, generalizing it to higher-order hierarchies of local sensitivity, e.g., local sensitivity of local sensitivity and further derivatives~\cite{KarwaEtAl2011,NguyenEtAl2023}, can be challenging, due to the privacy and failure probability cascading across multiple levels.
    \item Using Propose-Test-Release (PTR)~\cite{DworkLei2009} to propose a bound on local sensitivity to scale the temperature. If the proposed value is not sufficient for privacy, we might refuse to release an answer. 
    PTR for EM has been studied, but the main challenge (and its weakness) is its dependence on a proposed value in the PTR's simplest form, which seems difficult to estimate.
    \item Extending smooth sensitivity for EM. 
    Prior work observes that naive replacement of global sensitivity by smooth sensitivity in EM is not sufficient for privacy~\cite{FletcherIslam2017, ChavesEtAl2025}. 
    Therefore, the challenges of this approach include finding a flexible intermediate component that can adapt to both EM and smooth sensitivity for a unified solution, and the complexity of analyzing the regret bounds due to multiple complex components.
\end{enumerate}

\subsection{Contributions}
\label{sec:contrib}

Our main contributions are as follows:

\begin{itemize}
    \item We present explicit counterexamples for the direct applications of local and smooth sensitivities in setting the temperature scale in EM. 
    We show that local sensitivity requires the additive privacy parameter $\delta$ close to $1/2$, contrary to the principle that $\delta$ is ``cryptographically negligible''~\cite{Vadhan2017, DworkRoth2014}.
    For smooth sensitivity, a multiplicity construction denies any fixed range-independent approximate-DP at every fixed range-independent smoothness.
    \item We utilize the private upper bound on local sensitivity to design a conditional EM to achieve approximate DP.
    We then prove a recursive theorem for higher-order local sensitivity hierarchies, allowing privacy and regret bound to be propagated through all levels.
    \item Inspired by~\cite{LiNguyenVullikanti2024}, we binary search for a safe temperature scale using a $1$-Lipschitz lower bound on distance to unsafe datasets under the PTR framework.
    It addresses the current issue of the fixed scale PTR approach for EM that relies on an arbitrary ``proposed'' upper bound on local sensitivity, that is unclear how to set correctly and effectively.
    \item We give a smooth geometric envelope function that is admissible under the local dampening~\cite{FariasEtAl2023}. Distinctively, we design a logarithmic co-transformation of both utility score and smooth sensitivity, proved stable by a direct Lipschitz argument, before applying EM with transformed utility scores.
    Both methods achieve pure DP and preserve maximizers, i.e., candidates that maximize the utility scores.
    We derive explicit instance-dependent regret bounds in terms of the original utility score $q(x,r)$.
    \item We propose two alternative utilizations of smooth sensitivity.
    The first uses a categorical Gibbs law directly calibrated by smooth sensitivity that requires the smoothness to shrink with range size $m$, and privacy parameters ($\eps, \delta$).
    Our smoothness matches the theoretical lower bound required for privacy up to some constants depending only on $\eps$.
    The second route is by privatizing the smooth sensitivity.
    \item We provide rigorous regret bounds for all methods.
    Furthermore, we prove two theoretical lower bounds for temperature scales required for privacy: one for unequal scales of two neighboring datasets, the other for the required smoothness of direct application of smooth sensitivity.
    \item We apply the private upper bound on local sensitivity method, replacing the global sensitivity in EM, in solving problems to find a vertex (or an induced subgraph of size $k$) that is included in (or contains) the maximum number of copies of a motif $H$~\cite{everett1998analyzing, konar2022triangle, 10.1093/comnet/cnab052}.
    In the simplest form, $H$ is an $\ell$-clique.
    This family of problems have applications in many domains such as biomedical~\cite{10.3389/fphys.2012.00357}, document summarization~\cite{konar2022triangle}, and communication~\cite{10.1145/1557019.1557054}.
    We show that our high-probability regret bounds on Erd\H{o}s--R\'enyi graphs are significantly better than the worst case bounds when the graphs are sparse. 
\end{itemize}

\anil{some applications where these methods give improvement in performance or time would motivate this better}

Important privacy and utility properties of our results (in bold) are presented in Table~\ref{tab:map}.

\begin{table*}[ht]
\centering
\scriptsize
\setlength{\tabcolsep}{3pt}
\caption{Exact scope of the main selection guarantees.  A dash means that
the cited result is a privacy obstruction rather than a utility theorem.
For the shifted-Laplace construction, $\bar B_x(\xi)$ from \eqref{eq:one-level-upper} is one valid construction-specific choice of this family. $0<\rho<1$, $0<\delta<1$.
$b_x^\uparrow(\xi)$ is a finite deterministic bound such that $\Pr_{B\sim K_x}[B > b_x^\uparrow(\xi)] \leq \xi$.
}
\label{tab:map}
\begin{tabularx}{\linewidth}{@{}
>{\raggedright\arraybackslash}p{0.18\linewidth}
>{\raggedright\arraybackslash}p{0.27\linewidth}
>{\raggedright\arraybackslash}p{0.22\linewidth}
>{\raggedright\arraybackslash}X@{}}
\toprule
Method & Assumptions & Privacy & Regret and success probability \\
\midrule
Public global scale
& $0<b<\infty$ and $b\ge\GS_q$
& Pure $\varepsilon_s$
& $2b\log(m/\rho)/\varepsilon_s$; $1-\rho$ \\
\textbf{Direct local scale}
& Unprotected $\LS_q(x)$ denominator
& No uniform $(\varepsilon,\delta)$-DP with $\delta<1/2$;
Proposition~\ref{prop:naive-local}
& --- \\
\textbf{Direct smooth scale}
& Fixed $\beta,\lambda>0$, growing finite $m$
& No fixed range-independent $(\varepsilon_0,\delta_0)$-DP with
$\delta_0<1$; Theorem~\ref{thm:multiplicity}
& --- \\
\textbf{Private upper scale}
& Definition~\ref{def:private-upper} with $L=\LS_q$
& Equation~\eqref{eq:two-order}; if $K$ is pure,
$(\varepsilon_t+\varepsilon_s,\eta)$
& $2b_x^\uparrow(\xi)\log(m/\rho)/\varepsilon_s$;
$1-\rho-\xi$ \\
\textbf{Recursive hierarchy}
& Finite hierarchy \eqref{eq:hierarchy} and positive floors
& Equation~\eqref{eq:hier-final}
& Equation~\eqref{eq:hier-regret};
$1-\rho-\sum_j\xi_j$ \\
Fixed PTR
& Public $b>0$, $\varepsilon_t$-DP test, false pass at most $\delta$
& $(\varepsilon_t+\varepsilon_s,\delta)$
& $2b\log(m/\rho)/\varepsilon_s$; $1-\rho$ given pass \\
\textbf{Adaptive PTR}
& Finite safe fallback and certificates; for utility, monotonicity and
existence of $i_+(x)$
& $(\varepsilon_t+\varepsilon_s,\eta)$
& Equation~\eqref{eq:search-regret}; $1-\rho-\zeta$ \\
\textbf{Geometric dampening}
& Finite $\beta$-smooth $S\ge\LS_q$ and floor $s_0>0$
& Pure $\varepsilon$
& Equation~\eqref{eq:smooth-damp-regret}; $1-\rho$ \\
\textbf{Log dampening}
& Finite positive $\beta$-smooth $S_g\ge\LS_g$ and public $c>0$
& Pure $\varepsilon$
& Equation~\eqref{eq:log-regret}; $1-\rho$ \\
\textbf{Direct smooth Gibbs}
& Finite positive $\beta$-smooth $S_g\ge\LS_g$ and
calibration \eqref{eq:direct-calibration}
& $(\varepsilon,\delta)$
& $S_g(x)\log(m/\rho)/\lambda$; $1-\rho$ \\
\textbf{Private smooth scale}
& Finite positive $\beta$-smooth $S\ge\LS_q$
& $(\varepsilon_t+\varepsilon_s,\eta)$
& Use \eqref{eq:private-smooth-upper} in \eqref{eq:em-tail};
$1-\rho-\xi$ \\
\bottomrule
\end{tabularx}
\end{table*}

\subsection{Structure of the paper}

We present our results in the following structure.
Section~\ref{sec:prelims} (Prelims) presents the privacy model, standard tools for private selection: the Exponential mechanism and the Gibbs law, different notions of sensitivity, and utility measurement.
Section~\ref{sec:related} discusses existing studies in private selection, local sensitivity and its variants in private mechanism design.
Section~\ref{sec:failure} shows how straightforward applications of local sensitivity and smooth sensitivity fail to guarantee DP, and makes the claim that fixed smoothness smooth sensitivity cannot provide range-independent DP.
Section~\ref{sec:private-scale} presents the first valid calibration of EM using privatized upper bounds on local sensitivity, that can be extended to multiple layers of local sensitivities.
Section~\ref{sec:adaptive-ptr} describes an adaptive, private search for an upper bound on local sensitivity, utilizing efficient, tractable lower bound certificates for distances to high sensitivity datasets.
Section~\ref{sec:smooth} discusses several strategies to apply smooth sensitivity to calibrate EM.
The first two transformation methods target the utility score, replacing the original score by new scores with bounded or controllable global sensitivity, guaranteeing range-independent pure DP.
The third and fourth utilize the smoothness, either scaling it by the candidate range's size, or using it to calibrate its logarithmic transformation for privacy, yielding weaker approximate DP.
Section~\ref{sec:lower-bounds} discusses more quantitative challenges in using variants of local sensitivity, stating the dependency of the Gibbs law temperature of two neighbor datasets, and the maximum smoothness factor in terms of the range and privacy parameters to achieve DP.
Finally, Section~\ref{sec:apps} demonstrates using our methods to obtain significant improvements for a family of vertex/subgraph selection problems in some regimes of input.

\section{Preliminaries}
\label{sec:prelims}

\begin{definition}[Dataset graph]\label{def:dataset-graph}
A \emph{dataset graph} is an undirected graph $(\X,\sim)$.  Its vertices are
the possible private datasets, and an edge $x\sim x'$ means that the two
datasets are neighbors under the chosen privacy model.  Let $d(x,z)$ be the
shortest-path distance, with $d(x,z)=+\infty$ for different connected
components.
We define the distance to a set: for $\mathcal H\subseteq\mathcal X$, set $d(x,\mathcal H)=\inf_{z\in\mathcal H}d(x,z)$, with $\inf_\varnothing=+\infty$.
\end{definition}

The \emph{dataset graph} here refers to an abstract representation of all datasets of the dataset universe $\X$.
Each dataset is represented by a vertex.
Two vertices have an edge if their associated datasets are immediate neighbors defined by the privacy model.
This definition gives us a clean definition of distance between datasets without fixing a privacy model.
Intuitively, given two datasets $x,y \in \X$, $d(x,y)$ as defined in Def.~\ref{def:dataset-graph} is the minimum number of privacy units needed to transform one dataset to the other.
We note that this definition should not be mistaken for a privacy model on graphs, such as node- or edge-DP.
When a dataset itself is a graph, the distance $d(x,y)$ denotes the edit distance defined by the privacy model (node-DP or edge-DP) between two graphs $x$ and $y$, not the shortest path distance between two vertices of a specific input graph.

\anil{this is not tied to graphs right?}


\begin{definition}[Mechanism and differential privacy (DP)]\label{def:dp}
A randomized mechanism $M$ maps each dataset $x$ to a probability
distribution on an output space.  In probability language, this map is a
\emph{probability kernel}.  It is $(\varepsilon,\delta)$-differentially
private (approximate-DP) if, for every ordered edge $x\sim x'$ and measurable event $A$,
\[
 \Prb[M(x)\in A]\le
 \e^\varepsilon\Prb[M(x')\in A]+\delta.
\]
When $\delta=0$, the mechanism is \emph{pure $\varepsilon$-DP}.  Because the
graph is undirected, the same definition includes the reversed ordered pair.
\end{definition}

\textbf{Scores, sensitivity, Gibbs laws, and regret.}

Let $\R$ be a nonempty finite public candidate range of size
$m=|\R|\ge1$, and let $q:\X\times\R\to\mathbb R$ be a score to measure the utility of a candidate w.r.t. a specific dataset.  
The non-private selection target is any of $\argmax_{r\in\mathcal R}q(x,r)$, i.e., the maximizers.
The gap of a candidate $r$ is its utility gap to the maximizers' utility, i.e., $g(x,r)=\max_{r'} q(x,r')-q(x,r)\ge0$.

\anil{$q$ is just a score. maximization is a separate point}

For an edge $x\sim x'$, define its score oscillation $\omega_q(x,x')=\max_{r\in\R}|q(x,r)-q(x',r)|.$
Define element, dataset-local, and global sensitivity by
\begin{align*}
 \LS_q(x,r)&=\sup_{x'\sim x}|q(x,r)-q(x',r)|,\\
 \LS_q(x)&=\max_{r\in\R}\LS_q(x,r), \quad
 \GS_q=\sup_{x\in\X}\LS_q(x).
\end{align*}
For a scalar function $f:\X\to\mathbb R$, write
$\LS_f(x)=\sup_{x'\sim x}|f(x)-f(x')|$, with the supremum over an empty set, e.g., for an isolated $x$, defined to be $0$, and
$\GS_f=\sup_x\LS_f(x)$. 
These suprema take values in $[0,+\infty]$.  
Every mechanism below that uses
a sensitivity as a denominator or noise scale explicitly assumes the
relevant value is finite; with a public floor that lower-bounds a scale away from zero \anil{ceiling?}\dungnote{It is floor} supplies strict positivity where
needed.

For $a>0$, $\Lap(a)$ denotes the mean-zero Laplace distribution with density
$(2a)^{-1}\exp(-|z|/a)$; ``standard Laplace'' means $\Lap(1)$.

For a public temperature scale $b>0$ and selection privacy budget $\varepsilon_s>0$, define
\begin{align}
 Q_x^b(r)&=\frac{\exp\{\varepsilon_s q(x,r)/(2b)\}}
 {Z_x(b)}, \nonumber\\
 Z_x(b)&=\sum_{a\in\R}\exp\{\varepsilon_s q(x,a)/(2b)\}.\label{eq:gibbs}
\end{align}
This is called a categorical Gibbs or
softmax law.  
\textbf{The exponential mechanism (EM)} samples and releases a candidate $r$ with probability $Q_x^b(r)$.
The statistical-physics temperature is $2b/\varepsilon_s$.
The value $2$ is a temperature coefficient: its value can reduce to $1$ if $q$ is a monotonic function in $x$: under the standard add/remove adjacency with every candidate score monotone in the same direction across each ordered add/remove edge~\cite{DworkRoth2014}.
A larger $b$ flattens the law and worsens accuracy; a
smaller $b$ concentrates on high scores but is private only when justified.
Standard EM sets $b$ to $\GS_q$, when $0 < GS_q < \infty$.
When $GS_q = 0$, any $b>0$ is a valid choice, and we do not consider $GS_q = \infty$.

\begin{restatable}[Standard edgewise exponential-mechanism bound]{lemma}{lemedgeem}
\label{lem:edge-em}
If $x\sim x'$ and $\omega_q(x,x')\le b$, then the two laws in
\eqref{eq:gibbs} have likelihood ratio at most
$\e^{\varepsilon_s}$ in both directions, i.e., $Q_x^b(r) \le \e^{\eps_s}Q_{x'}^b(r)$ and $Q_{x'}^b(r) \le \e^{\eps_s}Q_{x}^b(r)$.
\anil{need to clarify}
\end{restatable}

The lemma is local to one edge.  A single global bound is sufficient because
it certifies every edge at once, but it is not logically necessary.  Later
mechanisms either certify the edge with high probability, charge the unsafe
probability to $\delta$, or replace $q$ by a utility score with bounded or controllable global sensitivity.

Let $q^*(x)=\max_r q(x,r)$.  The \emph{regret} of output $R$ is
\[
 \operatorname{Reg}_q(x,R)=q^*(x)-q(x,R)\ge0.
\]

\anil{could put as an observation}
\begin{observation}
The standard tail
calculation gives, conditional on any $b>0$,
\begin{equation}\label{eq:em-tail}
 \Prb_{R\sim Q_x^b}\!\left[
 \operatorname{Reg}_q(x,R)>
 \frac{2b}{\varepsilon_s}\log\frac m\rho
 \right]\le\rho,
 \quad 0<\rho<1.
\end{equation}
\end{observation}

\anil{would be useful to write this as an algorithm and give a name, so it can be used in all the results}
In this paper, the scale $b$ is not merely a publicly known quantity, as being set to $\GS_q$ as in the original EM.
$b$ can be the output of a (possibly randomized) function whose input includes $x$, the input private dataset.
A scale kernel $K_x$ denotes a randomized distribution $K(x, \cdot)$, in which a sampled $B\sim K_x$ is used to set the value of $b$. 
The EM with a random-scale mechanism is a joint kernel: we first sample $B\sim K_x$, then sample $R\sim Q_x^B$ conditionally on $(x,B)$.
In practice, we have to decide explicitly whether the released output is $(B,R)$ or only $R$.

\textbf{Smooth upper bounds.} 

\begin{definition}[$\beta$-smooth upper bound]
Fix $\beta\ge0$. A nonnegative function $S$ on the dataset space is a
\emph{$\beta$-smooth upper bound} on the local sensitivity of $f$ if
the following two conditions hold:
\begin{enumerate}
    \item For every dataset $x$,
        $S(x) \ge \LS_f(x).$
    \item For every pair of neighboring datasets $x\sim x'$,
        $S(x) \le \e^\beta S(x').$
    
\end{enumerate}
\end{definition}

The second condition implies, by chaining along a shortest sequence of
neighboring datasets, that for all datasets $x$ and $y$ that $d(x,y)<\infty$:
\[
    S(x) \le \e^{\beta d(x,y)} S(y).
\]

For a nonnegative function $L$ and $\beta\ge0$, define
\begin{equation}\label{eq:minimal-smooth}
 S_{L,\beta}(x)=
 \sup_{z:\,d(x,z)<\infty}\e^{-\beta d(x,z)}L(z).
\end{equation}
\anil{this will work even if $d(x,z)$ is unbounded, right?}
\dungnote{Maybe, I'm not 100\% sure.}
Vertices outside the
connected component of $x$ are excluded from the supremum.
Equation~\eqref{eq:minimal-smooth} is valid on components with unbounded diameters and can equal $+\infty$.
A positive-result mechanism using it assumes it is finite.
This possibly extended value is the pointwise smallest $\beta$-smooth upper
bound on $L$: it dominates $L$ and satisfies
$S_{L,\beta}(x')\le\e^\beta S_{L,\beta}(x)$ on every edge
~\cite{NRS2007}. 
It is pointwise minimal: for every $\beta$-smooth
upper bound $S$ and every dataset $x$, $ S_{L,\beta}(x) \le S(x)$.
\anil{might be good to state it as in Nissim's paper with the two conditions}
Minimality is useful for exact counterexamples and the tightest
bound.  
None of our positive-result privacy proofs requires the minimal choice.
\anil{what does affirmative mean?}
\dungnote{Affirmative means positive results, unlike the negative, counterexample results that uses the minimal form.}

\anil{if $x$ is a graph, it could have multiple connected components. Not clear what this means}
\dungnote{$x$ is a vertex in the dataset graph.}

\begin{table}[tb]
\caption{Frequently used notations.}
\begin{center}
\small
\begin{tabularx}{\linewidth}{@{}lX@{}}
\toprule
Symbol & Meaning \\
\midrule
$(\X,\sim),d$ & Dataset graph, adjacency, and extended graph distance \\
$\R,m$ & Public candidate  range and its size \\
$q,q^*,\operatorname{Reg}_q$ & Score, optimal score, and regret \\
$\omega_q,\LS_q,\GS_q$ & Edge oscillation, local sensitivity, global sensitivity \\
$Q_x^b,Z_x(b)$ & Gibbs/EM law at scale $b$ and its partition function \\
$\varepsilon_s,\varepsilon_t$ & Selection and scale/test privacy budgets \\
$S_{L,\beta}$ & Minimal $\beta$-smooth upper bound on $L$ \\
$b_0$ & Public positive floor for a random scale \\
\bottomrule
\end{tabularx}
\end{center}
\end{table}



\section{Related work}
\label{sec:related}

The exponential mechanism (EM) is an established standard method for private selection~\cite{DworkRoth2014, McSherryTalwar2007, MinamiEtAl2016, awan2019benefits}, where the private mechanism samples from  a set of public candidates, favoring ones with higher utility scores.
EM has been an important building block in various differentially private algorithms~\cite{nguyen2021differentially, mohamed2022differentially, imola2023differentially, nguyen2021kmeans, gillenwater2022joint}.
\cite{McSherryTalwar2007} proposes the canonical form of EM, using a common (to all candidates) public sensitivity to calibrate the temperature of the sample distribution.
\cite{raskhodnikova2016lipschitz} generalizes EM, using public per-candidate sensitivities to transform utility to have unit sensitivity before applying EM.
Our work targets a divergent setting: EM calibrated by an all-candidate dataset-dependent sensitivity.

Sensitivity is an important concept in differential privacy, that quantifies the change of output when the input dataset is changed slightly.
Sensitivity is used to calibrate noise injection for numeric queries~\cite{DworkLei2009, DworkRoth2014, Vadhan2017}.
EM uses sensitivity to set the temperature of candidate sampling distribution~\cite{McSherryTalwar2007}.

Unlike global sensitivity of a function which is the worst case sensitivity across possible input datasets, local sensitivity, which is measured at a specific input dataset, is not enough for privacy protection of numeric queries by itself~\cite{Vadhan2017, NRS2007}.
\cite{KarwaEtAl2011} utilizes second-order local sensitivity to construct a private upper bound on the first-order local sensitivity for some graph statistics.
\cite{LuMiklau2014} uses a repeated shifted Laplace upper bound through higher orders.
\cite{NguyenEtAl2023} further generalizes the concept to higher-order sensitivity hierarchies for private subgraph counting.

The Propose-Test-Release (PTR) framework~\cite{DworkLei2009} uses local sensitivity in a different way: it proposes a sensitivity bound, either publicly known or privately released (Propose phase), it calculates the (noisy) distance from the input dataset to the region of datasets whose local sensitivity is not bounded by the proposed value and checks if the distance is safe (Test phase) for privacy before finally injecting noise calibrated by the public bound or abstaining from answering (Release phase).
\cite{NRS2007} proposes the smooth sensitivity framework, which is a smooth upper bound on local sensitivity, in which their values at two neighboring datasets do not differ multiplicatively by more than a smoothness value $(\e^\beta)$.
Using the smooth sensitivity to calibrate a one-dimensional Cauchy noise gives pure DP, and a Laplacian noise yields approximate DP.
\cite{ZhangEtAl2015} constructs ladder functions that produce quality scores with unit global sensitivity ($GS=1$) from sensitivity at increasing distance from the input dataset.
\cite{FariasEtAl2023} shares similar ideas with the ladder function, defines an admissible sensitivity function indexed by radius, uses it to transform the utility score to have unit global sensitivity and applies the traditional EM to the transformed scores.
In this paper, we adapt this framework, using smooth sensitivity to construct admissible envelopes and derive their rigorous regret bounds.

Most local variants of sensitivity and their related quantities have been initially proposed for numeric queries.
There are few attempts at adapting them for EM~\cite{Duff2020, AsiDuchi2020, sakaue2026sensitivitydifferentialprivacymetric, durfee2025unifyinglaplacemechanisminstance}.
EM has been incorporated into the PTR framework~\cite{LiuKongOh2022,RedbergEtAl2023,RamsayChenouri2021}.
The fixed scale PTR~\cite{LiuKongOh2022, RamsayChenouri2021} relies on a known public bound $b^+$ on local sensitivity, and the calculation of the distance to high sensitivity datasets, which raise some practical issues that we will discuss further in Appendix~\ref{app:ptr}.
For numeric queries, \cite{LiNguyenVullikanti2024} proposes a noisy binary search over $b^+\in\{1,\ldots,\lceil GS_f\rceil\}$, assuming a finite public global sensitivity $GS_f$ or its upper bound, and replaces each threshold's possibly intractable exact distance to high sensitivity datasets by a threshold-dependent nonnegative, 1-Lipschitz lower certificate.

The 2021 addendum~\cite{DBLP:journals/corr/FletcherI16} to~\cite{FletcherIslam2017} acknowledges that its direct smooth sensitivity substitution lacks a formal privacy proof and reports an observation that the privacy may not always be guaranteed.
Subsequently, \cite{ChavesEtAl2025} formally forms a counterexample, claiming such substitution does not necessarily meet the pure DP standard.
Their claim fixes the temperature coefficient and the smoothness to the target privacy parameter as $\lambda=\eps/2=0.25$ and $\beta=\eps=0.5$.
Both results do not establish an impossibility uniformly over independently fixed smoothness and temperature coefficient.
The counterexamples do not rule out range-independent approximate DP guarantees for arbitrary $\delta<1$.
In this paper, we strengthen these observations by showing that, for every fixed smoothness $\beta>0$ and temperature coefficient $\lambda>0$, no $(\eps_0,\delta_0)$-DP guarantee with finite $\eps_0$ and $\delta_0<1$ can hold over all finite candidate ranges.

\textbf{Smooth sensitivity has been used in private selection}, though the following work does not directly use it to set the temperature of the categorical Gibbs law (Eq.~\eqref{eq:gibbs}).
\cite{LiuTalwar2019} studies a more general formulation of private selection, in which each candidate is a DP mechanism, with noise calibrated by smooth sensitivity as one possible instantiation.
The framework utilizes an oracle model, using a random number of oracle calls, instead of sampling from the categorical Gibbs law.
It can improve the utility bound when a utility-maximizing candidate has low smooth sensitivity.
\cite{MinamiEtAl2016} studies the selection problem using a related continuous Gibbs-posterior under Euclidean, convex/Lipschitz-loss, and strongly log-concave-prior conditions, not the finite-candidate selection problem.
\cite{AntonovaEtAl2025} addresses the problem along a different axis, using public candidate-specific global sensitivities instead of our dataset-specific sensitivity approach.
It also shows that naive independent score perturbation calibrated separately to these sensitivities need not satisfy DP.

There is a family of approaches called perturb-and-select~\cite{DworkRoth2014}, that first privatizes the utility score and selects the candidates, either by argmax or sampling.
Cauchy noise generated by smooth sensitivity can be applied directly to the vector of utility scores of all candidates~\cite{NRS2007},
followed by an argmax or Gibbs sample from the noisy vector which serves as post-processing with no extra incurred privacy cost.
In order for noise to be added to each coordinate, it requires a norm-specific local sensitivity analysis and the calibration depends on candidate range.
\cite{YamamotoShibuya2024} recently uses smooth sensitivity, but also with range-dependent calibration to privatize the utility scores.
\cite{HladikTetek2025} improves the smooth noise for single real-valued queries.
\cite{ChavesEtAl2025} proposes smooth noisy max, utilizing a common smooth upper bound on local sensitivity, admitting some families of noise, demonstrating its effectiveness over global sensitivity on various tasks: percentile selection, decision tree, and random forest.


\section{Why direct data-dependent scales fail}
\label{sec:failure}

The mechanisms in this section have the same structure as the Gibbs law in
\eqref{eq:gibbs}, but replace the public scale $b$ by a deterministic function of
the dataset, namely the local sensitivity $\LS_q(x)$ or the smooth sensitivity
$S_{\LS_q,\beta}(x)$, while keeping the utility score $q$ intact.
We call these methods ``direct'' because these replacement functions are in their original forms, directly using their output as scales without any privacy guarantee.
\anil{not sure if the negative results are well motivated, beyond just mentioning (or putting in the SI). 
Is there connection with any other literature, which makes these versions interesting?}
Our Proposition~\ref{prop:naive-local} (Naive local-scale substitution fails) is an exponential mechanism's analogue of~\cite{Vadhan2017} (Section 3)'s and~\cite{NRS2007} (Example 1)'s arguments for numeric queries.
For smooth sensitivity, an invalid substitution of global sensitivity has been known~\cite{ChavesEtAl2025}, but the observation is limited to a fixed smoothness $\beta = \epsilon = 0.5$ and a fixed temperature coefficient ($2$, the same as global sensitivity's coefficient), under the pure DP model.
Our discussion of the impossibility of smooth sensitivity provides a rigorous, generalized theorem:
no $(\eps,\delta)$-DP with $\delta<1$ is possible at any fixed $\beta$ independent from candidate range size $m$, and at any fixed temperature coefficient $\lambda$ (Theorem~\ref{thm:multiplicity}).

\subsection{Local sensitivity}

\begin{proposition}[Naive local-scale substitution fails]
\label{prop:naive-local}
For every $\varepsilon>0$, there is no uniform
$(\varepsilon,\delta)$-DP guarantee with $\delta<1/2$ for
\[
 N_\varepsilon(x)(r)\propto
 \exp\!\left\{\frac{\varepsilon q(x,r)}{2\LS_q(x)}\right\},
\]
even when the denominator is positive at every dataset used by the
mechanism.
\end{proposition}

\begin{proof}
Let the dataset graph be the path $z\sim x\sim x'\sim w$ and let
$\R=\{a,b\}$.  For $G>0$ and $M>1$, set
\[
\begin{array}{c|rrrr}
 &z&x&x'&w\\ \hline
q(\cdot,a)&G+1&G&G&G+M\\
q(\cdot,b)&0&0&0&0.
\end{array}
\]
Then $\LS_q(x)=1$ and $\LS_q(x')=M$, although the score vectors at the edge $(x,x')$ are identical.  Hence:
\begin{align*}
     \Prb[N_\varepsilon(x)=b]=\frac1{1+\e^{\varepsilon G/2}};\;
 \Prb[N_\varepsilon(x')=b]=\frac1{1+\e^{\varepsilon G/(2M)}}.
\end{align*}

Taking $M=G^2$ and $G\to\infty$, the first probability tends to zero and the
second to $1/2$.  The ordered event $\{b\}$ from $x'$ to $x$ therefore
requires $\delta\to1/2$.
\end{proof}

The leakage comes entirely from the non-private scale values set by the local sensitivity.
Since the local sensitivities also take into account other adjacent edges to each of the endpoints of $x\sim x'$,
the difference between their respective values on the endpoints can be arbitrarily large, even when their utility vectors are identical.

\subsection{Smooth sensitivity}

\begin{proposition}[Naive smooth scale is not pure DP]
\label{prop:naive-smooth}
Fix $\varepsilon>0$ and $\beta>0$.  The direct rule
\[
 P_x(r)\propto
 \exp\!\left\{\frac{\varepsilon q(x,r)}
 {2S_{\LS_q,\beta}(x)}\right\}
\]
need not be $\varepsilon$-DP.
\end{proposition}

\begin{proof}
Use the path $x\sim y\sim z$, candidates $\{a,b\}$, and scores
\[
 q(x,\cdot)=q(y,\cdot)=(0,L),\qquad q(z,\cdot)=(A,L),
\]
where $A,L>0$.  The local sensitivities at the three vertices are $0,A,A$.
The exact minimal smooth bounds at $x$ and $y$ are therefore
$\e^{-\beta}A$ and $A$.  Substituting the values into the EM, we have:
\[
 \frac{P_y(a)}{P_x(a)}=
 \frac{1+\exp(\e^\beta \eps L/(2A))}{1+\exp(\eps L/(2A))},
\]
which is unbounded as $L\to\infty$. 
\end{proof}

The analysis above works for every public positive $\beta$ and the exact minimal smooth upper bound.
In this analysis, smooth sensitivity does not provide privacy, because the probability ratio of returning the same candidate at an edge cannot be bounded by the smoothness: while the ratio between $e^\beta C$ and $C$ is bounded by $e^\beta$ regardless of $C$, the ratio between $\exp(e^\beta C)$ and $\exp(C)$ is unbounded when $C\to \infty$.
The next Theorem provides a stronger impossibility result, as it denies every approximate DP guarantee for fixed temperature coefficient and fixed smoothness.
\anil{needs some clarification}

\begin{restatable}[Multiplicity amplification]{theorem}{thmMultiplicity}
\label{thm:multiplicity}
Fix a coefficient $\lambda>0$, smoothness $\beta>0$, target
$\varepsilon_0<\infty$, and $0\le\delta_0<1$.  There are a finite candidate
range, a score, and neighboring vertices $x,y$ with identical score vectors
such that
\[
 M(x)(r)\propto
 \exp\{\lambda q(x,r)/S_{\LS_q,\beta}(x)\}
\]
violates $(\varepsilon_0,\delta_0)$-DP.  The two exact minimal smooth scales
have ratio $\e^\beta$ on the edge.
\end{restatable}

\begin{proof}
Use the path $w\sim x\sim y\sim z$ and candidates
$\R=\{0,1,\ldots,k\}$.  Choose $B>\e^\beta$ and set
\[
 q(x,\cdot)=q(y,\cdot)=(C,0,\ldots,0),
\]
while $q(w,1)=1$, $q(z,1)=B$, and every other outer score equals its middle
value.  The local sensitivities along the path are $(1,1,B,B)$.  Thus the
minimal smooth scales at $x,y$ are $B\e^{-\beta}$ and $B$.
The probability of high-score candidates $\Pr[M(y) = \R_0] \propto \exp(\lambda C / B)$ and $\Pr[M(x) = \R_0]\propto \exp(\lambda C / (Be^{-\beta}))$.
Specifically:

\begin{align*}
    \Pr[M(x) = \R_0] &= \frac{\exp(\frac{\lambda C}{Be^{-\beta}})}{k + \exp(\frac{\lambda C}{Be^{-\beta}})}\\
    \Pr[M(y) = \R_0] &= \frac{\exp(\frac{\lambda C}{B)}}{k + \exp(\frac{\lambda C}{B})}\\
\end{align*}

Put $t=\lambda C/B \ge 0$, choose $1<a<\e^\beta$, and take
$k=\lfloor\e^{at}\rfloor$.  For the event $A$ that one of the $k$
low-score candidates, i.e., not $\R_0$, is selected,
\[
 \Prb_y[A]=\frac{k}{k+\e^t} 
 \ge \frac{1}{1 + 2\e^{t(1-a)}}
 \longrightarrow1,
 \qquad
 \Prb_x[A]=\frac{k}{k+\e^{\e^\beta t}}
 \le \frac{1}{1 + \e^{t(e^\beta - a)}}
 \longrightarrow0
 \quad(t\to\infty).
\]
For every fixed pair of target parameters $(\eps_0,\delta_0)$, a sufficiently large but finite integer $k$ rejects the privacy claim.
\end{proof}

The theorem sets fixed values for $\beta$ and $\lambda$ independently of $m,\eps, \delta$. 
It does not rule out a sufficient calibration in which $(\beta,\lambda)$ are calculated based on  $(m,\eps,\delta)$ as we propose in Section~\ref{sec:smooth-others}.

\section{Private upper scales and sensitivity hierarchies}
\label{sec:private-scale}
 
Using a privatized one-sided upper bound on local sensitivity to calibrate Laplacian noise is known for numeric queries~\cite{KarwaEtAl2011,KasiviswanathanEtAl2013,Vadhan2017}; but has not been used to calibrate the Gibbs law in EM.
\anil{selection is a different type of question, not setting for Laplacian?}
The formulations below are its interfaces with a conditional exponential mechanism, i.e.,
the Gibbs law with its scale being conditioned on a randomized value that is privately released.
First, we compute $B$, a private value that is lower-bounded by $\LS_q(x)$ with probability at least $1-\eta$.
EM calibrated by $B$ is conditionally private by Lemma~\ref{lem:edge-em}.
The failure probability that $B$ underestimates $\LS_q(x)$ is accounted for by the $\delta$-part of DP.
The most challenging analysis is the privacy bound, as sampled from the Gibbs law, must be taken over all possible values of the privatized local sensitivity, a numeric query to produce $B$ that possibly has a continuous domain.

\begin{definition}[Private upper scale]\label{def:private-upper}
Fix a public floor $b_0>0$.  A scale kernel $K_x$ supported on
$[b_0,\infty)$ is an
$(\varepsilon_t,\delta_t;\eta)$ private upper bound for a nonnegative
finite-valued function $L$ if $K$ is $(\varepsilon_t,\delta_t)$-DP and
\[
 \Prb_{B\sim K_x}[B<L(x)]\le\eta
 \quad\text{for every }x.
\]
\end{definition}

\anil{notation and definition not clear.
scale kernel?
What is $K$, and a distribution being private? And does the subscript mean anything specific?}
\dungnote{I Put a definition in Prelims}
The public floor $b_0$ avoids division by zero.  
Using a small $b_0$ allows a more accurate distribution when sensitivity is low.
The floor does not fix a bad case where $B < \LS_q(x)$, and the privacy guarantee consumes it under the failure probability of approximate DP, i.e., in the $\delta$-part of DP.  

Given $B=b$, let $Q_x^b$ be a conditional EM such that, on
an edge $x\sim x'$, $Q_x^b$ and $Q_{x'}^b$ are mutually
$\varepsilon_s$-private whenever $b\ge \LS_q(x)$.
This is because $b$ only needs to cover one of the endpoints, since both $\LS_q(x)$
and $\LS_q(x')$ cover the oscillation across $x\sim x'$.

Then we will prove that the joint release
$(B,R)$, where $B\sim K_x$ and $R\sim Q_x^B$ is
\begin{equation}\label{eq:two-order}
\left(\varepsilon_t+\varepsilon_s,
 \min\{\e^{\varepsilon_s}\delta_t+\eta,
        \delta_t+\e^{\varepsilon_t}\eta,1\}\right)\text{-DP}.
\end{equation}  

\begin{restatable}[Private-upper-scale exponential selection]{theorem}{thmPrivateScaleEM}
\label{thm:private-scale-em}
If $K$ is an $(\varepsilon_t,\delta_t;\eta)$ private upper bound for
$\LS_q$, then releasing $(B,R)$ with $R\sim Q_x^B$ has the guarantee in
\eqref{eq:two-order}.  If $K$ is pure $\varepsilon_t$-DP, the guarantee is
$(\varepsilon_t+\varepsilon_s,\eta)$-DP.  Conditional on $B=b$, the regret
bound is \eqref{eq:em-tail}.
\anil{some terms need to be clarified. 
what does releasing $(B, R)$ mean-- is this as in (\ref{eq:gibbs}), with $B$ instead of $b$? $b$ is said to be public there, so this might be confusing}
\dungnote{I defined it in Prelims now.}
\end{restatable}

\begin{proof}
Let $K_x$ be an $(\varepsilon_t,\delta_t;\eta)$ private upper bound for
$L$.  Given $B=b$, let $A_x^b$ be a conditional output kernel such that, on
an edge $x\sim x'$, $A_x^b$ and $A_{x'}^b$ are mutually
$\varepsilon_s$-private whenever $b\ge L(x)$.  

Then we will prove that the joint release
$(B,Y)$, where $B\sim K_x$ and $Y\sim A_x^B$, is
\begin{equation}
\left(\varepsilon_t+\varepsilon_s,
 \min\{\e^{\varepsilon_s}\delta_t+\eta,
        \delta_t+\e^{\varepsilon_t}\eta,1\}\right)\text{-DP}.
\end{equation}  

Fix an ordered edge $x\sim x'$ and a joint event $E$.  
Let $E_b$ be its
section at scale $b$, and let $G_x=\{b:b\ge L(x)\}$.  
On $G_x$, conditional
privacy gives $A_x^b(E_b)\le\e^{\varepsilon_s}A_{x'}^b(E_b)$, while
$K_x(G_x^c)\le\eta$, where $G_x^c$ is the complement of $G_x$.  The expectation form of approximate DP for $K$ gives
\begin{align*}
 \Prb_x[(B,Y)\in E]
 &= \int A_x^b(E_b)K_x(db) = \int_{G_x^c}A_x^b(E_b)K_x(db) + \int_{G_x} A_x^b(E_b)K_x(db) \\
 &\le \eta+  \int_{G_x} A_x^b(E_b)K_x(db) \le \eta + \e^{\varepsilon_s} \int_{G_x}A_{x'}^b(E_b)K_x(db)\\
 & \le \eta + \e^{\varepsilon_s + \eps_t}\int_{G_x}A_{x'}^b(E_b)K_{x'}(db) +\e^{\varepsilon_s}\delta_t
 \le \e^{\varepsilon_t+\varepsilon_s}
       \Prb_{x'}[(B,Y)\in E]
       +\eta+\e^{\varepsilon_s}\delta_t.
\end{align*}

For the other order, first transport the bounded function
$b\mapsto A_x^b(E_b)$ from $K_x$ to $K_{x'}$, and then split the target
integral on $G_{x'}$.  This yields
\begin{align*}
 \Prb_x[(B,Y)\in E]  &= \int A_x^b(E_b)K_x(db) \le \e^{\eps_t}\int A_x^b(E_b)K_{x'}(db) +\delta_t \\
 &\le \e^{\eps_t}\int_{G_{x'}} A_x^b(E_b)K_{x'}(db) + \e^{\eps_t}\int_{G_{x'}^c} A_x^b(E_b)K_{x'}(db) + \delta_t 
 \le  \e^{\eps_t}\int_{G_{x'}} A_x^b(E_b)K_{x'}(db) + \e^{\eps_t}\eta + \delta_t \\
 &\le  \e^{\eps_t + \eps_s}\int_{G_{x'}} A_{x'}^b(E_b)K_{x'}(db) + \e^{\eps_t}\eta + \delta_t
 \le\e^{\varepsilon_t+\varepsilon_s}
       \Prb_{x'}[(B,Y)\in E]
       +\delta_t+\e^{\varepsilon_t}\eta.
\end{align*}

Taking $L=\LS_q$ and $A_x^b=Q_x^b$ is valid by Lemma~\ref{lem:edge-em}, the Theorem follows.
\end{proof}




Applying the Laplace mechanism ($W$ in the below Corollary) using the global sensitivity of local sensitivity, we then provide a one-level shifted-Laplace scale.
The shifting $C/{\varepsilon_t}\log\frac1{2\eta}$ is needed to guarantee that the scale is a one-sided high-probability upper bound on the true local sensitivity, as $\Pr[Lap(C/\eps_t) < - C/{\varepsilon_t}\log\frac1{2\eta}] \le \eta$ by properties of the Laplacian distribution.

\begin{corollary}[Shifted-Laplace local scale]
\label{cor:one-level}
Suppose $\LS_q$ is finite-valued.  Let $C=\GS_{\LS_q}<\infty$, fix $\varepsilon_t>0$ and
$0<\eta<1/2$.  If $C>0$, draw $W\sim\Lap(C/\varepsilon_t)$ and set
\begin{equation}\label{eq:one-level-scale}
 B=\max\left\{b_0,\LS_q(x)+W+
 \frac C{\varepsilon_t}\log\frac1{2\eta}\right\}.
\end{equation}
Then private-upper-scale EM is
$(\varepsilon_t+\varepsilon_s,\eta)$-DP.  If $C=0$, the deterministic
scale $\max\{b_0,\LS_q(x)\}$ is constant on edges and gives pure
$\varepsilon_s$-DP.
\anil{if $C=0$ this is basically like uniform distribution. Maybe you can assume that's not the case in the beginning itself.}
\dungnote{When $C=0$, $\LS_q$ is constant along every edge and hence on each connected component; the mechanism uses a constant scale. It is not generally uniform over candidates unless all scores tie}
\end{corollary}

\begin{proof}
The scalar function $\LS_q$ has global sensitivity $C$.  The Laplace release
in \eqref{eq:one-level-scale} is pure $\varepsilon_t$-DP, and its lower tail
underestimates $\LS_q(x)$ with probability at most $\eta$.  Applying
Theorem~\ref{thm:private-scale-em}, the Corollary follows.
\end{proof}

If $0<\xi<1/2$, then with probability at least $1-\xi$,
\begin{equation}
\label{eq:one-level-upper}
\begin{aligned}
B \le \bar B_x(\xi)
&:= \max\Biggl\{b_0,
 \LS_q(x)+\frac{C}{\varepsilon_t}
 \left(\log\frac{1}{2\eta}+\log\frac{1}{2\xi}\right)
 \Biggr\}.
\end{aligned}
\end{equation}

Combining \eqref{eq:one-level-upper} with \eqref{eq:em-tail} gives the 
regret in Table~\ref{tab:map} with probability $1-\rho-\xi$.

\textbf{Recursive higher-order private scales.} 
Corollary~\ref{cor:one-level} uses the first-order local sensitivity.  
Karwa et al. go one derivative further for graph statistics, using the local sensitivity of local sensitivity (second-order),
and Nguyen et al. develop application-specific higher-order hierarchies~\cite{KarwaEtAl2011,NguyenEtAl2023}.  
The next theorem gives a general selection interface for such a hierarchy.

\begin{definition}[Sensitivity hierarchy]\label{def:hierarchy}
A depth-$k$ sensitivity hierarchy for $q$ consists of finite nonnegative
functions $H_0,\ldots,H_k$ and a finite public $C>0$ such that
\begin{equation}\label{eq:hierarchy}
 H_0\ge\LS_q\;;\;
 H_{j+1}\ge\LS_{H_j}\ (0\le j<k)\;;\;
 C\ge\GS_{H_k}.
\end{equation}
\end{definition}
When all iterates are finite and the top one has a finite cap, the minimal
choice is exact local sensitivities $H_0=\LS_q$ and $H_{j+1}=\LS_{H_j}$.  
\anil{upper proxy, uncertified heuristic-- not clear}
\dungnote{Rewritten below.}
The hierarchy does not necessarily consist of layers of local sensitivities. 
Each $H_j$ is a deterministic, possibly efficiently computable, upper bound that satisfies the inequality in \eqref{eq:hierarchy} for every possible dataset. 
A loose bound may affect the utility, but does not impact the privacy analysis.
Any approximation must be accompanied by rigorous analyses of privacy and underestimation, where the failure probability is taken into account.

Fix floors $b_j>0$, privacy allocations $\varepsilon_j>0$, and
$0<\eta_j<1/2$ for $j=0,\ldots,k$.  Set $B_{k+1}=C$.  Independently for
$j=k,k-1,\ldots,0$, draw standard Laplace $Z_j$ and set
\begin{equation}\label{eq:recursive-scale}
 B_j=\max\left\{b_j,H_j(x)+\frac{B_{j+1}}{\varepsilon_j}
 \left(Z_j+\log\frac1{2\eta_j}\right)\right\}.
\end{equation}
Finally, release $(B_0, R)$ with a sampled $R\sim Q_x^{B_0}$.

\begin{theorem}[Recursive private-scale selection]
\label{thm:hierarchy-privacy}
Put $\eta_{k+1}=0$ and
$\bar\varepsilon_{k+1}=\bar\delta_{k+1}=0$.  For
$j=k,k-1,\ldots,0$, define
\begin{align}
 \bar\varepsilon_j&=\bar\varepsilon_{j+1}+\varepsilon_j,
 \label{eq:hier-eps}\\
 \bar\delta_j&=\min\left\{
 \e^{\varepsilon_j}\bar\delta_{j+1}+\eta_{j+1},
 \bar\delta_{j+1}+\e^{\bar\varepsilon_{j+1}}\eta_{j+1},1
 \right\}.
 \label{eq:hier-delta}
\end{align}
Then the joint release $(B_0,R)$ is
\begin{equation}\label{eq:hier-final}
\left(\bar\varepsilon_0+\varepsilon_s,
 \min\{\e^{\varepsilon_s}\bar\delta_0+\eta_0,
        \bar\delta_0+\e^{\bar\varepsilon_0}\eta_0,1\}
\right)\text{-DP}.
\end{equation}
A simpler, possibly looser, final additive term is
\begin{equation}\label{eq:hier-simple}
 \delta\le\sum_{j=0}^k\e^{\bar\varepsilon_j}\eta_j.
\end{equation}
The same privacy statement holds if all internal scales $B_j, j\in[k]$ are also released.
\end{theorem}

\begin{proof}
Let $T_j=(B_j,\ldots,B_k)$.  For every $j$, the one-sided Laplace tail gives
\begin{equation}\label{eq:hier-coverage}
 \Prb[B_j<H_j(x)]\le\eta_j,
\end{equation}
independently of the value or safety of $B_{j+1}$; the public floor $b_j$ can only
reduce this probability.

At the top, $B_{k+1}=C\ge\GS_{H_k}$ is a deterministic safe scale.  Given
$B_{j+1}=b$, the next line of \eqref{eq:recursive-scale} is a shifted
Laplace release of $H_j(x)$ with scale $b/\varepsilon_j$, followed by public
clipping.  It is conditionally $\varepsilon_j$-DP whenever
$b\ge\LS_{H_j}(x)$.  Coverage of $H_{j+1}$ implies this good-scale event by
\eqref{eq:hierarchy}.  
Applying the statement at Equation~\eqref{eq:two-order} inductively to
$T_{j+1}$ gives exactly \eqref{eq:hier-eps}--\eqref{eq:hier-delta}.
Applying the lemma once more to $T_0$ and conditional EM gives
\eqref{eq:hier-final}.  Repeatedly choosing the second term in
\eqref{eq:hier-delta} yields
$\bar\delta_j\le\sum_{i=j+1}^k
\e^{\bar\varepsilon_i}\eta_i$; the final second accounting order then gives
\eqref{eq:hier-simple}.
\end{proof}

\begin{restatable}[Recursive high-probability regret]{theorem}{thmHierarchyUtility}
\label{thm:hierarchy-utility}
Fix $0<\xi_j<1/2$, set $U_{k+1}(x)=C$, and define downward

\begin{align}
U_j(x)
&= \max\Biggl\{b_j,
H_j(x)+\frac{U_{j+1}(x)}{\varepsilon_j}
\left(\log\frac{1}{2\eta_j}+\log\frac{1}{2\xi_j}\right)
\Biggr\}.
\label{eq:hier-U} 
\end{align}

For every $0<\rho<1$, with probability at least
$1-\rho-\sum_{j=0}^k\xi_j$,
\begin{equation}\label{eq:hier-regret}
 \operatorname{Reg}_q(x,R)
 \le\frac{2U_0(x)}{\varepsilon_s}\log\frac m\rho.
\end{equation}
\end{restatable}

\begin{proof}
Except with probability $\sum_j\xi_j$, every upper Laplace tail satisfies
$Z_j\le\log(1/(2\xi_j))$.  We show downward that $B_j\le U_j$.  If the
shifted factor $Z_j+\log(1/(2\eta_j))$ is negative, then its product with the
positive $B_{j+1}/\varepsilon_j$ is nonpositive.  Hence the affine term in
\eqref{eq:recursive-scale} is at most $H_j(x)\le U_j(x)$, while also
$b_j\le U_j(x)$; taking their maximum gives $B_j\le U_j$.  Otherwise it is
safe to substitute $B_{j+1}\le U_{j+1}$, and \eqref{eq:hier-U} follows.  On this event
$B_0\le U_0$; combine with the conditional EM tail \eqref{eq:em-tail}.
\end{proof}

This hierarchy can be strictly better or worse than one-level calibration.
It helps only when the local upper levels and final cap make $U_0$ smaller
than \eqref{eq:one-level-upper}.  
Each level replaces its global sensitivity by a private high probability upper bound on its own local sensitivity, consuming privacy and adding up failure probability.
The theorem is therefore a reusable interface, not a claim that higher order always means lower noise.

\section{Adaptive certificate-search PTR selection}
\label{sec:adaptive-ptr}

The simplest variant of PTR is supplied with a public $b$ (the propose phase), followed by a test if the noisy distance to $\{z: \LS_q(z) > b\}$, i.e., the distance to the highly sensitive zone (or unsafe zone where local sensitivities of datasets are greater than $b$), is safe for a release of $Q_x^b$.
In the unsafe case, PTR abstains from answering.
Incorporating PTR into EM is straightforward, but it has two major issues: we need to provide a public $b$ or privately estimate it, which consumes extra privacy, and the noisy distance described above is not always tractable.
Setting a too low value for $b$ may make PTR refuse to answer often, while setting a too high value for $b$ imposes a high temperature which increases regret bounds.
In this section, we adapt the ideas from~\cite{LiNguyenVullikanti2024} to address both issues.
We privately search over a public grid of values for a good $b$, and we replace the distance to the highly sensitive zone by a tractable $1$-Lipschitz lower certificate.
Each application might require a distinct lower certificate design, which depends on the nature of the problem and its settings; its computational efficiency is not universal.

Assume a finite public global cap $b_J$ on $\GS_q$ is known.  Let
$J\ge2$ and $0<b_1<\cdots<b_J$ be public, with $b_J\ge\GS_q$.  Define
\[
 \mathcal H_i=\{z:\LS_q(z)>b_i\},\qquad
 D_i(x)=d(x,\mathcal H_i).
\]
Distance to a set is 1-Lipschitz wherever finite.  We do not require exact
distance.

\begin{definition}[Finite lower certificate]\label{def:certificate}
A certificate family consists of finite, nonnegative, publicly specified
functions $\phi_i:\X\to\mathbb R$ satisfying
\begin{equation}\label{eq:certificate}
 0\le\phi_i(x)\le D_i(x);\quad
 |\phi_i(x)-\phi_i(x')|\le1\quad(x\sim x').
\end{equation}
For the utility theorem, also require
$\phi_1(x)\le\cdots\le\phi_J(x)$ for every $x$.
\end{definition}

The certificate $\phi$ may underestimate the true distance to the unsafe zone, but never overestimate it.
The utility depends on the monotonicity of $\phi$ (w.r.t. $i$), while privacy does not.
A distance truncated at a finite public threshold is a valid $1$-Lipschitz certificate, even when the distance $D_i (x) = \infty$.
The terminal bound $b_J$ is a safe upper bound, and when all others fail, we can use $b_J$ without evaluating its certificate $\phi_J$.
Since the bad sets are nested, i.e., $\mathcal \mathcal H_{i+1} \subseteq H_{i}$, the exact distances are nondecreasing in terms of $i$.


\anil{this is the description of algorithm Adaptive certificate-search, so would be useful to name}

Fix $\varepsilon_t,\varepsilon_s>0$ and $0<\eta<1/2$.  Let
\begin{equation}\label{eq:search-parameters}
 L=\lceil\log_2J\rceil,\qquad
 s=\frac L{\varepsilon_t},\qquad
 \tau=s\log\frac L{2\eta}.
\end{equation}

\textbf{Algorithm.} After defining the certificate family, we then run a private lower-bound binary search for the lowest index that passes the noisy test $\phi_i(x)+\Lap(s)>\tau$ and use its bound $b_i$ to calibrate the Gibbs law for EM.
We start with $\ell = 1, h = J$.
Let $i = \lfloor (\ell + h) / 2\rfloor$, $W\sim\Lap(s)$, test $\phi_i(x)+W>\tau$. 
If the test is true, we say that it passes the test and reduce the high end $h=i$.
If the test is false, we increase the low end $\ell = i + 1$.
We repeat the test until $\ell \geq h$.
If all $i<J$ fail, we use $b_J$ as a fallback without having to test it. 
Since $b_J \geq \GS_q$, the privacy follows the guarantee of EM with global sensitivity.


\begin{restatable}[Adaptive certificate-search privacy]{theorem}{thmSearchPrivacy}
\label{thm:search-privacy}
The selected scale $B$ is pure $\varepsilon_t$-DP and satisfies
\begin{equation}\label{eq:search-coverage}
 \Prb[B<\LS_q(x)]\le\eta.
\end{equation}
Consequently, the joint output $(B,R)$, and hence $R$ alone, is
$(\varepsilon_t+\varepsilon_s,\eta)$-DP.
\end{restatable}

\begin{proof}
The search makes at most $L$ adaptive tests.  Each noisy certificate is
$\varepsilon_t/L$-DP, so adaptive basic composition and post-processing give
pure $\varepsilon_t$-DP for $B$.

Every returned index below $J$ was explicitly passed.  If its scale is
unsafe, $b_i<\LS_q(x)$, then $x\in\mathcal H_i$, so
$D_i(x)=\phi_i(x)=0$.  Since $\tau\ge0$, the conditional false-pass
probability at that step is
\[
 \Prb[W>\tau]=\tfrac12\e^{-\tau/s}=\eta/L.
\]
A union bound over at most $L$ adaptive positions proves
\eqref{eq:search-coverage}.  The untested fallback is safe.

For completeness, fix an ordered edge $x\sim x'$ and a joint event $A$.
Write $A_i=\{r:(b_i,r)\in A\}$ and
$K_x(i)=\Prb_x[B=b_i]$.  Splitting off the bad scale mass and applying
edgewise EM gives
\begin{align*}
 \Prb_x[(B,R)\in A]
 &\le\eta+
 \sum_{i:b_i\ge\LS_q(x)}K_x(i)Q_x^{b_i}(A_i)\\
 &\le\eta+\e^{\varepsilon_s}
 \sum_{i:b_i\ge\LS_q(x)}K_x(i)Q_{x'}^{b_i}(A_i)\\
 &\le\eta+\e^{\varepsilon_t+\varepsilon_s}
 \Prb_{x'}[(B,R)\in A].
\end{align*}
The last line is the expectation form of pure DP for the scale kernel,
applied to the bounded function
$i\mapsto\ind\{b_i\ge\LS_q(x)\}Q_{x'}^{b_i}(A_i)$.
Exchanging $x,x'$ proves the reverse ordered direction.
\end{proof}

\textbf{Utility Analysis.} For $0<\zeta<1/2$, put
\begin{equation}\label{eq:search-margin}
 u=s\log\frac L{2\zeta},\quad
 i_+(x)=\min\{i:\phi_i(x)\ge\tau+u\},
\end{equation}
when this set is nonempty.

\begin{proposition}[Certificate-dependent regret]
\label{prop:search-utility}
With a nonempty $\{i:\phi_i(x)\ge\tau+u\}$, if the certificates are nondecreasing in the grid index, then
\[
 \Prb[B\le b_{i_+(x)}]\ge1-\zeta.
\]
Consequently, for every $0<\rho<1$, with probability at least
$1-\zeta-\rho$,
\begin{equation}\label{eq:search-regret}
 \operatorname{Reg}_q(x,R)\le
 \frac{2b_{i_+(x)}}{\varepsilon_s}\log\frac m\rho.
\end{equation}
\end{proposition}

\begin{proof}
Conditional on any search history that queries an index $i\ge i_+(x)$,
monotonicity gives $\phi_i(x)\ge\tau+u$.  The fresh noise makes that query
fail with probability at most
$\tfrac12\e^{-u/s}=\zeta/L$.  If no such failure occurs, binary search cannot
move its lower endpoint above $i_+(x)$, so its output is at most that index.
A union bound proves the scale statement.  Conditional on any
$B\le b_{i_+(x)}$, the EM tail \eqref{eq:em-tail} at the larger scale
$b_{i_+(x)}$ is at most $\rho$; another union bound gives
\eqref{eq:search-regret}.

\end{proof}
\section{Smooth Sensitivity}
\label{sec:smooth}

In previous sections, we replace the public scale $b$ by privacy-protected functions, hoping that their values are not as large as the public value $\GS_q$ while keeping the original utility score $q(x,r)$.
In this section, we explore a different variant of local sensitivity, namely the smooth sensitivity.
The smooth sensitivity is an upper bound on local sensitivity at every dataset $x$, together with a constraint of its smoothness: the ratio of its values on $x\sim x'$ is upper bounded by $e^\beta$.
Section~\ref{sec:failure} confirms that straightforward application of smooth sensitivity does not provide a privacy guarantee.
In this section, we explore several different approaches.
The first and second methods (Section~\ref{sec:smooth-dampening} and Section~\ref{sec:log-transform}) utilize smooth sensitivity to transform the original utility score to ones with bounded or controllable global sensitivity.
The third and fourth methods (Section~\ref{sec:smooth-others}) utilize the smoothness factor, but in different ways: scaling the smoothness by privacy and range statistics, or privatizing its logarithmic function by the smoothness.
The first two methods show some advantages, by providing range-independent pure DP while retaining the maximizing candidates.

The first method is based on the local dampening framework.
Appendix~\ref{app:smooth} presents the rigorous analysis and further discussions of the adapted framework.
The dampening framework provides the privacy template but still requires an admissible envelope and a utility translation, which we use a geometric transformation of the smooth sensitivity to design.

The separate, independent method combines smooth sensitivity and the original utility score in the form of a logarithmic transformation to generate a new score.
The global sensitivity of the new score can be parameterized by changing some constant factors of the transformation, allowing us to achieve a desired public regret bound.

\subsection{A geometric transformation of smooth sensitivity}
\label{sec:smooth-dampening}

\textbf{The dampening framework.} We adapt the local dampening framework~\cite{FariasEtAl2023}.
They proposed an admissible sensitivity envelope (Definition~\ref{def:envelope}) that any satisfied functions can be applied to produce pure-DP by the same analysis.
Some existing admissible envelopes use candidate-dependent functions, that may change the rankings of candidates.
Also, \cite{FariasEtAl2023} do not provide a rigorous quantifiable instance-specific regret bound (except an instance-agnostic one that is not worse than the bound of the EM with the global sensitivity).
We use the smooth sensitivity to design our envelopes, that are independent of candidates, hence preserve the maximizers $\argmax_{r\in \R} q(\cdot,r)$.
We propose a geometric smooth envelope and provide rigorous regret bounds in terms of the original score $q(x,r)$.
We note that the smooth sensitivity alone does not satisfy the admissible envelope properties, as an edge-varying (might change across some $x\sim x'$), time-independent (no dependency on $t$ below) smooth sensitivity with $\beta > 0$ violates the second condition of Definition~\ref{def:envelope}.

\begin{definition}[Admissible sensitivity envelope]\label{def:envelope}
For each $x\in\X$, $r\in\R$, and integer $t\ge0$, let
$\gamma_x(t,r)>0$.  The family $\gamma$ is \emph{admissible} for a score
$u$ if
\begin{align}
 \gamma_x(0,r)&\ge \LS_u(x,r),
 \label{eq:envelope-zero}\\
 \gamma_x(t+1,r)&\ge\gamma_{x'}(t,r)
 &&(x\sim x',\ t\ge0),
 \label{eq:envelope-shift}
\end{align}
and $\sum_{t\ge0}\gamma_x(t,r)=+\infty$ for every $(x,r)$.
\end{definition}

For every integer $k$, define the signed cumulative breakpoints
\begin{equation}\label{eq:breakpoints}
 b_x(k,r)=
 \begin{cases}
  \sum_{j=0}^{k-1}\gamma_x(j,r),&k>0,\\
  0,&k=0,\\
  -\sum_{j=0}^{-k-1}\gamma_x(j,r),&k<0.
 \end{cases}
\end{equation}
Since $b_x(k,r)$ is only defined at integers $k$, we use piecewise-linear interpolation through
$(k,b_x(k,r))_{k\in\mathbb Z}$ to define $B_x(K,r)$ that takes any real number $K$.
Given a real number $K$, let $k: k \leq K < k + 1$, $B_x(K,r) = (k+1-K)b_x(k,r) + (K-k) b_x(k+1, r)$.
Positivity ($\gamma_x(t, r) > 0)$ and divergence ($\sum_{t} \gamma_x(t,r) = +\infty$) make the map $B_x(\cdot, r)$ a strictly increasing bijection and invertible.  
Write $T_x(\cdot,r)=B_x(\cdot,r)^{-1}$, i.e., inverse of $B_x: T_x(B, r) = T \iff B_x(T, r) = B$ and set
\begin{equation}\label{eq:dampened-score}
 \widetilde u(x,r)=T_x(u(x,r),r).
\end{equation}

\begin{restatable}[Admissible-envelope dampening
~\cite{FariasEtAl2023}]{theorem}{thmEnvelope}
\label{thm:envelope}
If $\gamma$ is admissible for $u$, then
\[
 |\widetilde u(x,r)-\widetilde u(x',r)|\le1
 \quad(x\sim x',\ r\in\R).
\]
Consequently, sampling with probability proportional to
$\exp\{\varepsilon\widetilde u(x,r)/2\}$ is pure
$\varepsilon$-DP.
\end{restatable}

\textbf{The geometric transformation of smooth sensitivity.}
Let us revisit the gap of a candidate, i.e., the distance to the maximizers' utility:

\begin{equation}\label{eq:centered-gap}
 g(x,r)=q^*(x)-q(x,r)\ge0;\;
v(x,r)=-g(x,r)
\end{equation}
Centering by $q^*(x)$ does not change a fixed-dataset Gibbs distribution.
It does change sensitivity, and the following factor two is necessary:
\begin{equation}\label{eq:gap-ls}
 \LS_g(x):=\max_r\sup_{x'\sim x}|g(x,r)-g(x',r)|
 \le2\LS_q(x).
\end{equation}

The factor of two is due to $|g(x,r)-g(x',r)| = |q^*(x) - q^*(x') + q(x', r) - q(x,r)| \leq |q^*(x) - q^*(x') | + |q(x, r) - q(x',r)| \leq 2LS_q(x)$.

Let $S$ be any finite $\beta$-smooth upper bound on $\LS_q$:
\begin{equation}\label{eq:S-smooth}
 S(x)\ge\LS_q(x),\qquad
 S(x')\le\e^\beta S(x)\quad(x\sim x').
\end{equation}
The minimal choice is $S_{\LS_q,\beta}$ from
\eqref{eq:minimal-smooth}.  Fix a public $s_0>0$ and write
$S_0(x)=\max\{S(x),s_0\}$; this maximum preserves
\eqref{eq:S-smooth}.  Define the candidate-independent geometric envelope
\begin{equation}\label{eq:geometric-envelope}
 \gamma_x(t,r)=2\e^{\beta t}S_0(x),\qquad t=0,1,\ldots.
\end{equation}
The factor two is needed to upper bound $LS_g(x)$. 
For $z\ge0$, let $C_x$ be the increasing piecewise-linear function whose
integer values are
\begin{equation}\label{eq:C-breakpoints}
 C_x(0)=0,\quad
 C_x(k)=2S_0(x)\sum_{j=0}^{k-1}\e^{\beta j}\quad(k\ge1),
\end{equation}
and let $T_x^+=C_x^{-1}$ on $[0,\infty)$.  The transformed score is
\begin{equation}\label{eq:smooth-score}
 h(x,r)=-T_x^+(g(x,r)).
\end{equation}

By the result of Theorem~\ref{thm:envelope}, $GS_h \leq 1$ that makes EM using the score $h(\cdot, r)$ and inverse temperature $\eps/2$ a valid private mechanism.
The intuition behind this transformation is that each $C_x(k)$ for an integer $k$ is viewed as the start of a new step in a ladder.
Function $T_x^+$ therefore picks a (fractional) step given a height at the value of $g(x,r)$.
Moving from $x$ to $x'$ changes $S_0$ by at most a factor of $e^\beta$.
Therefore, the new ladder is shifted by at most one step above or one step behind.
Since the smallest width of a step, i.e., the distance to the next step, is at least $2S_0 \geq LS_g$,
moving $g(x,r)$ to $g(x',r)$ changes at most $LS_g(x) \leq 2S_0$ in the height, which is at most one step width of $x$ ladder.
It follows that $g(x',r)$'s step in $x'$ ladder is at most $1$ unit away from $g(x,r)$'s step in $x$ ladder.


\begin{theorem}[Smooth-dampened exponential selection]
\label{thm:smooth-dampening}
Sample $R$ with probability proportional to
$\exp\{\varepsilon h(x,r)/2\}$.  This mechanism is pure
$\varepsilon$-DP and $h(x,\cdot)$ has exactly the same maximizers as
$q(x,\cdot)$.  For every $0<\rho<1$, with probability at least $1-\rho$,
\begin{equation}\label{eq:smooth-damp-regret}
 \operatorname{Reg}_q(x,R)\le C_x(K),
 \qquad K=\frac2\varepsilon\log\frac m\rho.
\end{equation}
In particular,
\begin{equation}\label{eq:smooth-damp-explicit}
 \operatorname{Reg}_q(x,R)\le
 \begin{cases}
  2S_0(x)\dfrac{\e^{\beta\lceil K\rceil}-1}{\e^\beta-1},
      &\beta>0,\\[2mm]
  2S_0(x)\lceil K\rceil,&\beta=0.
 \end{cases}
\end{equation}
\end{theorem}

\begin{proof}
Equation~\eqref{eq:gap-ls} makes the zeroth envelope value at least the
local sensitivity of $v=-g$.  For an edge $x\sim x'$,
\[
 \gamma_x(t+1,r)=2\e^{\beta(t+1)}S_0(x)
 \ge2\e^{\beta t}S_0(x')=\gamma_{x'}(t,r).
\]
Thus the envelope is admissible, and Theorem~\ref{thm:envelope} proves pure
privacy.  Because the same strictly increasing map is used for every
candidate at a fixed dataset, $h=-T_x^+(g)$ is maximized exactly when
$g=0$, so it preserves all maximizers.

At least one transformed score is zero and every transformed score is
nonpositive.  The standard EM tail gives
$\Prb[h(x,R)<-K]\le m\e^{-\varepsilon K/2}=\rho$.  On the complementary
event, monotonicity of $C_x$ gives $g(x,R)\le C_x(K)$.
Equation~\eqref{eq:smooth-damp-explicit} follows by rounding $K$ upward and
summing the geometric series.  The second line is its limit at $\beta=0$.
\end{proof}



\subsection{A logarithmic co-transformation of smooth sensitivity and utility score}
\label{sec:log-transform}

The local dampening framework transforms the utility score to a new score with unit sensitivity.
Inspired by the idea, we work on combining both smooth sensitivity and the utility score to generate a new score with controlled global sensitivity, without the admissible envelope.
The result is a logarithmic co-transformation as defined below.

Suppose $S_g>0$ is a finite
$\beta$-smooth upper bound on $\LS_g$; one may take $S_g=2S_0$.  For a
public tuning constant $c>0$, define
\begin{equation}\label{eq:log-score}
 h_c(x,r)=-c\log\!\left(1+\frac{g(x,r)}{cS_g(x)}\right).
\end{equation}

The co-transformation utility score has two significant advantages.
The first is its global sensitivity is controlled by two factors: the smoothness $\beta$ and the constant $c$, which can be adjusted to achieve desirable regret.
The second is the preservation of maximizers, in which candidates who maximize the original utility score also maximize the transformed score and have the most mass in the output distribution.

\begin{theorem}[Log co-transformation smooth selection]\label{thm:log-dampening}
The score $h_c$ has global sensitivity at most $1+c\beta$.  Therefore
sampling with weights
\begin{equation}\label{eq:log-weights}
 \exp\!\left\{\frac{\varepsilon h_c(x,r)}{2(1+c\beta)}\right\}
\end{equation}
is pure $\varepsilon$-DP and preserves all maximizers of $q$.  For $0 < \rho < 1$, with
probability at least $1-\rho$, its regret is at most
\begin{equation}\label{eq:log-regret}
 cS_g(x)\left[
 \left(\frac m\rho\right)^{2(1+c\beta)/(c\varepsilon)}-1
 \right].
\end{equation}
\end{theorem}

\begin{proof}
Put $\phi(z,s)=c\log(1+z/(cs))$.  For $s>0 , z\ge0$,
\[
 |\partial_z\phi(z,s)|\le1/s,
 \qquad |\partial_{\log s}\phi(z,s)|\le c.
\]
Across $x\sim x'$, insert the intermediate value
$\phi(g(x',r),S_g(x))$.  The first derivative bound and
$|g(x,r)-g(x',r)|\le S_g(x)$ cost at most one.  The second bound and
$|\log S_g(x)-\log S_g(x')|\le\beta$ cost at most $c\beta$.
Thus $\GS_{h_c}\le1+c\beta$, and ordinary EM proves privacy.
Strict monotonicity in $g$ preserves maximizers.  Applying the EM tail to
the transformed regret $-h_c$ and solving
\eqref{eq:log-score} for $g$ gives \eqref{eq:log-regret}.
\end{proof}

The constant $c$ controls both the score compression and the temperature of the Gibbs law.
A very small $c$ reduces the difference between a good candidate and a bad candidate strongly, which can give bad regret.
However, a large $c$ when $\beta$ is also large can set a high temperature for the Gibbs law, due to an increasing in $\GS_h$, which is also bad for the regret.
In practice, a desired $c$ can be chosen by optimizing the regret bound stated in \eqref{eq:log-regret} over a grid of $c$.


\subsection{Other approaches for smooth sensitivity}
\label{sec:smooth-others}

Furthermore, we provide two other methods in which smooth sensitivity can be used for temperature calibration, instead of utility transformations as above.
The first method scales $\beta$ and temperature coefficient $\lambda$ (in Equation~\eqref{eq:direct-gibbs}) as a function of $m, \eps, \delta$.
It does not violate the impossibility result presented in Section~\ref{sec:failure} that requires fixed $\beta$ and $\lambda$.
The second method privatizes smooth sensitivity, followed by a shifting trick as in Section~\ref{sec:private-scale}.

\textbf{Direct smooth sensitivity calibration.}
Let $S_g>0$ be a finite $\beta$-smooth upper bound on $\LS_g$.  For a
public coefficient $\lambda>0$, define the normalized-gap Gibbs law
\begin{equation}\label{eq:direct-gibbs}
 p_x(r)=\frac{\exp\{-\lambda g(x,r)/S_g(x)\}}
 {\sum_{a\in\R}\exp\{-\lambda g(x,a)/S_g(x)\}}.
\end{equation}

\begin{restatable}[Direct finite-range smooth Gibbs]{theorem}{thmDirectSmooth}
\label{thm:direct-smooth}
For $T\ge0$, put
\begin{equation}\label{eq:direct-parameters}
 C_\beta(T)=\e^\beta+(\e^\beta-1)T,
 \qquad a_T=m\e^{-\lambda T}.
\end{equation}
The law \eqref{eq:direct-gibbs} is
\begin{equation}\label{eq:direct-general-dp}
 \left(2\lambda C_\beta(T),
 (1+\e^{\lambda C_\beta(T)})a_T\right)\text{-DP}.
\end{equation}
In particular, for target $\varepsilon>0$ and $0<\delta<1$, let
\begin{equation}\label{eq:direct-calibration}
 \ell=\log\frac{m(1+\e^{\varepsilon/2})}{\delta},
 \beta\le\log\!\left(1+\frac{\varepsilon}{4\ell}\right),
 \lambda=\frac{\varepsilon}{4\e^\beta}.
\end{equation}
Then \eqref{eq:direct-gibbs} is $(\varepsilon,\delta)$-DP.    
\end{restatable}

\begin{proof}
Write $z_x(r)=g(x,r)/S_g(x)$, $w_x(r)=\e^{-\lambda z_x(r)}$, and
$Z_x=\sum_r w_x(r)$.  Since an optimal candidate has gap zero, $Z_x\ge1$.
For an edge $x\sim x'$, smoothness and the two endpoint upper bounds give
\begin{equation}\label{eq:normalized-edge}
 \frac{S_g(x')}{S_g(x)}\in[\e^{-\beta},\e^\beta],
 \qquad
 |g(x,r)-g(x',r)|\le\min\{S_g(x),S_g(x')\}.
\end{equation}
If $z_x(r)\le T$, these constraints imply
\begin{equation}\label{eq:normalized-good}
 |z_x(r)-z_{x'}(r)|\le C_\beta(T).
\end{equation}
For example,
$z_{x'}(r)\le\e^\beta(z_x(r)+1)$ gives the larger signed bound
$(\e^\beta-1)T+\e^\beta$; the reverse sign is no larger.

Split the target normalizer of $x': Z_{x'} = \sum_r \e^{-\lambda z_{x'}(r)}$ at $z_{x'}(r)=T$.  Apply
\eqref{eq:normalized-good} with $z_{x'} \leq T$ (good part); on its bad part ($z_{x'}(r) > T$), each
weight is at most $\e^{-\lambda T}$.  Thus
\begin{equation}\label{eq:normalizer-split}
 Z_{x'}\le\e^{\lambda C_\beta(T)}Z_x+a_T.
\end{equation}
For an event $A$, w.r.t. dataset $x$, candidates with $z_x(r)>T$ have total
probability at most $a_T$.  On the remaining candidates,
\eqref{eq:normalized-good} and \eqref{eq:normalizer-split} give
\begin{align*}
 p_x(A\cap\{z_x\le T\})
 &\le\e^{\lambda C_\beta(T)}
       \frac{\sum_{r\in A}w_{x'}(r)}{Z_x}\\
 &\le\e^{2\lambda C_\beta(T)}p_{x'}(A)
       +\e^{\lambda C_\beta(T)}a_T.
\end{align*}

Adding the bad part $p_x(z_x > T) \le a_T$:
\begin{align*}
    p_x(A) &\leq p_x (A \cap \{z_x \le T\}) + p_x(z_x > T)\\
     &\le\e^{2\lambda C_\beta(T)}p_{x'}(A)
       +\e^{\lambda C_\beta(T)}a_T + a_T\\
       &=\e^{2\lambda C_\beta(T)}p_{x'}(A)
       +(\e^{\lambda C_\beta(T)} + 1)a_T,
\end{align*}
proves \eqref{eq:direct-general-dp}. 
Exchanging $x,x'$ proves the other ordered direction.
Under \eqref{eq:direct-calibration}, take $T=\ell/\lambda$.  Then
$a_T=\delta/(1+\e^{\varepsilon/2})$ and
\[
 \lambda C_\beta(T)=\frac\varepsilon4
       +(\e^\beta-1)\ell\le\frac\varepsilon2.
\]
Substitution into \eqref{eq:direct-general-dp} proves the claim.
\end{proof}

\begin{restatable}[Direct finite-range regret]{theorem}{thmDirectSmoothRegret}
\label{thm:direct-smooth-regret}

    For every dataset $x$, let $R\sim p_x$, where $p_x$ is the
normalized-gap Gibbs law in \eqref{eq:direct-gibbs}.
Then, for every $t\ge 0$,
\begin{equation}
\label{eq:direct-regret-tail}
 \Prb\!\left[\operatorname{Reg}_q(x,R)>t\right]
 \le
 \min\!\left\{
 1,
 m\exp\!\left(-\frac{\lambda t}{S_g(x)}\right)
 \right\}.
\end{equation}
Consequently, for every $0<\rho<1$, with probability at least $1-\rho$,
\begin{equation}
\label{eq:direct-regret}
 \operatorname{Reg}_q(x,R)
 \le
 \frac{S_g(x)}{\lambda}\log\frac{m}{\rho}.
\end{equation}
Under the $(\varepsilon,\delta)$-DP calibration in
\eqref{eq:direct-calibration}, this becomes
\begin{equation}
\label{eq:direct-calibrated-regret}
 \operatorname{Reg}_q(x,R)
 \le
 \frac{4\e^\beta S_g(x)}{\varepsilon}
 \log\frac{m}{\rho}
\end{equation}
with probability at least $1-\rho$.
\end{restatable}

\begin{proof}
Fix $x$ and write
\[
 Z_x
 =
 \sum_{a\in\R}
 \exp\!\left\{
 -\frac{\lambda g(x,a)}{S_g(x)}
 \right\}.
\]
Because $\R$ is finite and nonempty, some $r^*\in\R$ maximizes
$q(x,\cdot)$.
Hence $g(x,r^*)=0$, so its Gibbs weight is one and therefore $Z_x\ge1$.
Moreover,
\[
 \operatorname{Reg}_q(x,r)
 =
 q^*(x)-q(x,r)
 =
 g(x,r).
\]
Thus, for every $t\ge0$,
\begin{align*}
 \Prb\!\left[\operatorname{Reg}_q(x,R)>t\right]
 &=
 \frac{1}{Z_x}
 \sum_{r:\,g(x,r)>t}
 \exp\!\left\{
 -\frac{\lambda g(x,r)}{S_g(x)}
 \right\} \\
 &\le
 m\exp\!\left\{
 -\frac{\lambda t}{S_g(x)}
 \right\}.
\end{align*}
Combining this inequality with the trivial probability bound $1$ proves
\eqref{eq:direct-regret-tail}.
Setting
\[
 t
 =
 \frac{S_g(x)}{\lambda}\log\frac{m}{\rho}
\]
makes the exponential upper bound equal to $\rho$, proving
\eqref{eq:direct-regret}.
Finally, \eqref{eq:direct-calibration} sets
$\lambda=\varepsilon/(4\e^\beta)$.
Substituting this value into \eqref{eq:direct-regret} gives
\eqref{eq:direct-calibrated-regret}.
\end{proof}

\textbf{Privatized smooth sensitivity.}
Suppose $S>0$ satisfies \eqref{eq:S-smooth}.  Then
$\log S$ has global sensitivity at most $\beta$.  Fix
$\varepsilon_t>0$ and $0<\eta<1/2$.  If $\beta>0$, draw
$Y\sim\Lap(\beta/\varepsilon_t)$; if $\beta=0$, take $Y=0$.  Set
\begin{equation}\label{eq:private-smooth-scale}
 B=\max\!\left\{b_0,
 S(x)\exp\!\left(Y+\frac\beta{\varepsilon_t}
 \log\frac1{2\eta}\right)\right\}.
\end{equation}

\begin{restatable}[Private smooth upper scale]{theorem}{thmPrivateSmooth}
\label{cor:private-smooth}
Releasing $B$ from \eqref{eq:private-smooth-scale} and then sampling
$R\sim Q_x^B$ is $(\varepsilon_t+\varepsilon_s,\eta)$-DP.  For every
$0<\xi<1/2$, with probability at least $1-\xi$,
\begin{equation}\label{eq:private-smooth-upper}
 B\le\max\!\left\{b_0,
 S(x)\exp\!\left[\frac\beta{\varepsilon_t}
 \left(\log\frac1{2\eta}+\log\frac1{2\xi}\right)\right]
 \right\}.
\end{equation}
Combining this event with \eqref{eq:em-tail} gives the corresponding regret
bound with probability at least $1-\xi-\rho$.
\end{restatable}

\begin{proof}
For $\beta>0$, the noisy log scale is pure $\varepsilon_t$-DP, and the
one-sided Laplace tail makes $B<S(x)$ occur with probability at most $\eta$.
Since $S\ge\LS_q$, Theorem~\ref{thm:private-scale-em} applies.  When
$\beta=0$, $S$ is constant on every graph component and the release is
zero-DP.  The upper tail of $Y$ gives \eqref{eq:private-smooth-upper}.
\end{proof}

\section{Lower bounds}
\label{sec:lower-bounds}


In this section, we present two lower bounds for setting the temperature scale of EM for privacy in two settings, one for arbitrary and deterministic scales for any pair of neighboring datasets, and the other for scales set by smooth sensitivity.
It shows that our direct smooth sensitivity calibration in Section~\ref{sec:smooth-others} is tight, matching the required upper bound for the smoothness up to some constants depending on $\eps$.

\begin{theorem}[Unequal-scale converse]
\label{thm:reversal-scale-tradeoff}
Let $x\sim y$, let $b_x,b_y>0$, and consider the deterministic-scale Gibbs
laws
\[
 P_D(r)=\frac{\exp\{\lambda q(D,r)/b_D\}}
 {\sum_{s\in\R}\exp\{\lambda q(D,s)/b_D\}},
 \qquad D\in\{x,y\},
\]
where $\lambda>0$.  Suppose two candidates $a,b$ reverse preference across
the edge with margins
\begin{align*}
 q(x,a)-q(x,b)&=\Delta_x>0,\\
 q(y,b)-q(y,a)&=\Delta_y>0.
\end{align*}
If the two laws satisfy pure $\varepsilon_0$-DP on this edge, then
\begin{equation}
 \frac{\lambda\Delta_x}{b_x}+
 \frac{\lambda\Delta_y}{b_y}\le 2\varepsilon_0.
 \label{eq:reciprocal-scale-tradeoff}
\end{equation}
\end{theorem}

\begin{proof}
Pure DP applied to the answers $a$ and $b$, in the two ordered directions,
gives
\[
 \frac{P_x(a)}{P_y(a)}\le \e^{\varepsilon_0},
 \qquad
 \frac{P_y(b)}{P_x(b)}\le \e^{\varepsilon_0}.
\]
Multiplying them gives $\frac{P_x(a)P_y(b)}{P_y(a)P_x(b)} \le e^{2\eps_0}$, that cancels both partition functions ${\sum_{s\in\R}\exp\{\lambda q(D,s)/b_D\}}\;,\; D\in\{x,y\}$. 
 The logarithm of the left
side is exactly
$\lambda\Delta_x/b_x+\lambda\Delta_y/b_y$, proving
\eqref{eq:reciprocal-scale-tradeoff}.  
\end{proof}

The pure DP gives us a clean condition, while extending to approximate DP requires a data-dependent term in the r.h.s..
When $P_x(a), P_y(b) > \delta_0$, the r.h.s. yields $2\eps_0+2\log\frac{\min(P_x(a), P_y(b))}{\min(P_x(a), P_y(b))-\delta_0}$.
For simplicity, Theorem~\ref{thm:reversal-scale-tradeoff} focuses on the pure DP setting.

Theorem~\ref{thm:reversal-scale-tradeoff} states that under the assumption that for a pair of neighbor datasets $x\sim y$, there exist two candidates such that each is more preferable at exactly one dataset in $\{x, y\}$, their dataset-specific scales $b_x$ and $b_y$ must be jointly constrained.
This supports why naive local sensitivity fails, since $LS_q(x)$ and $LS_q(y)$ can differ arbitrarily, and gives an intuition for why our proposed methods work: they all provide some constraints on $b_x$ and $b_y$: fixed-scaled PTR sets $b_x$ and $b_y$ to the same value, while our proposed methods, i.e., adaptive PTR, private upper bound on local sensitivity and direct smooth sensitivity, informally control how much $b_x$ and $b_y$, or their distributions, differ via different mechanisms.


The next Theorem sets the upper bound required for the smoothness in the direct application of smooth sensitivity as a temperature scale, as in Section~\ref{sec:smooth-others}.

\begin{theorem}[Necessary smoothness at fixed privacy]
\label{thm:unbalanced-smooth-necessity}
Fix $\lambda>0$, $0<\delta_0\le 1/4$, and $m\ge3$.  Put $k=m-1$ and
$r=\e^\beta$.  If the direct rule \eqref{eq:direct-gibbs}, using the minimal
smooth upper bound $S_g=S_{\LS_g,\beta}$ with $\forall x, S_g(x) > 0$, is
$(\varepsilon_0,\delta_0)$-DP uniformly over finite gap scores on $m$
candidates, then
\begin{equation}
 r-1\le
 \frac{\log\!\left(2(\e^{\varepsilon_0}-\delta_0)/(1-2\delta_0)\right)}
 {\log\!\left(k(1-2\delta_0)/(2\delta_0)\right)}.
 \label{eq:joint-md-necessary}
\end{equation}
Consequently, for fixed $\varepsilon_0$,
\begin{equation}
 \beta=O\!\left(\frac{1}{\log(m/\delta_0)}\right).
 \label{eq:joint-md-rate}
\end{equation}
In particular, when both privacy parameters are fixed, one must have
$\beta=O(1/\log m)$ as the range grows.
\end{theorem}

\begin{proof}
Consider the path $w\sim x\sim y\sim z$.  At the two middle vertices let
the gap vectors be $g(x,\cdot)=g(y,\cdot)=(0,C,\ldots,C)$.  At $w$, increase
one positive gap by $1$; at $z$, increase it by $B$, where $B>\e^\beta$.
The local sensitivities along the path are $(1,1,B,B)$, so minimality in
\eqref{eq:minimal-smooth} gives
\[
 S_g(x)=B\e^{-\beta},\qquad S_g(y)=B.
\]
Let $A$ be the set of the $k$ positive-gap candidates and put
$t=\lambda C/B$.  Then
\[
 P_y(A)=\frac{k\e^{-t}}{1+k\e^{-t}},
 \qquad
 P_x(A)=\frac{k\e^{-rt}}{1+k\e^{-rt}}.
\]
Choose $t=\log\{k(1-2\delta_0)/(2\delta_0)\}$, so that
$P_y(A)=2\delta_0$.  Writing
\[
 H=k^{r-1}\left(\frac{1-2\delta_0}{2\delta_0}\right)^r,
\]
we have $P_x(A)=1/(1+H)$.  DP in the direction $y$ to $x$ gives
$2\delta_0\le \e^{\varepsilon_0}/(1+H)+\delta_0$, hence
$H\le \e^{\varepsilon_0}/\delta_0-1$.  Taking logarithms and isolating
$r-1$ proves \eqref{eq:joint-md-necessary}.  Finally,
$\beta=\log r\le r-1$, and the denominator in
\eqref{eq:joint-md-necessary} has order $\log(m/\delta_0)$.
\end{proof}

Theorem~\ref{thm:unbalanced-smooth-necessity} strengthens the impossibility result of Section~\ref{sec:failure},  which states that a $\beta$ that is independent of $m$ is not possible under fixed $(\eps,\delta)$-DP. 
Our direct smooth sensitivity sets $\beta\le \log\left(1 + \frac{\eps_0}{4\log{(m(1+e^{\eps_0/2})/\delta_0)}}\right) = O\left(\frac{\eps_0}{\log(m/\delta_0)}\right)$ for fixed $\eps_0, \delta_0$ in Equation \eqref{eq:direct-calibration}, that matches the required upper bound up to some constants depending on $\epsilon_0$. 

\section{Applications}
\label{sec:apps}

\begin{table*}[hb]
\centering
\scriptsize
\setlength{\tabcolsep}{3pt}
\caption{Worst-case global-sensitivity regret and high-probability regret on
$G(n,p), 0 < p < 1$.  The last column is the ratio of our bound to the GS bound; hence a
quantity $o(1)$ is an explicit asymptotic improvement. 
$O_\ell, O_H$ denote asymptotic bounds treating $\ell$ and $H$-related factors as constants.
Write $K_n=K_m$ for $m=n$ and $K_{n,k}=K_m$ for $m=\binom nk$.
}
\label{tab:app-graph-summary}
\begin{tabularx}{\linewidth}{@{}p{0.18\linewidth}p{0.18\linewidth}
p{0.33\linewidth}X@{}}
\toprule
Problem & GS regret & Our $G(n,p)$ regret & Asymptotic gap \\
\midrule
$K_\ell$ participation vertex
& $\Theta_\ell(K_n n^{\ell-2})$
& $\begin{aligned}
 O_\ell\!\bigl(K_n[&D^{\ell-2}
 +n^{\ell-3}\Lambda/\varepsilon]\bigr)
 \end{aligned}$
& $\begin{aligned}
 O_\ell\!\bigl(&(D/n)^{\ell-2}+\Lambda/(\varepsilon n)\bigr)
 \end{aligned}$ \\
$K_\ell$-densest $k$-set
& $\Theta_\ell(K_{n,k}k^{\ell-2})$
& $\begin{aligned}
 O_\ell\!\bigl(K_{n,k}[&\min\{k,D\}^{\ell-2}+k^{\ell-3}\Lambda/\varepsilon]\bigr)
 \end{aligned}$
& $\begin{aligned}
 O_\ell\!\bigl(&(\min\{k,D\}/k)^{\ell-2}+\Lambda/(\varepsilon k)\bigr)
 \end{aligned}$ \\
$H$-participation vertex
& $\Theta_H(K_n n^{h-2})$
& $\begin{aligned}
 O_H\!\bigl(K_n[&D^{h-2}+n^{h-3}\Lambda/\varepsilon]\bigr)
 \end{aligned}$
& $\begin{aligned}
 O_H\!\bigl(&(D/n)^{h-2}+\Lambda/(\varepsilon n)\bigr)
 \end{aligned}$ \\
$H$-densest $k$-set
& $\Theta_H(K_{n,k}k^{h-2})$
& $\begin{aligned}
 O_H\!\bigl(K_{n,k}[&\min\{k,D\}^{h-2}+k^{h-3}\Lambda/\varepsilon]\bigr)
 \end{aligned}$
& $\begin{aligned}
 O_H\!\bigl(&(\min\{k,D\}/k)^{h-2}+\Lambda/(\varepsilon k)\bigr)
 \end{aligned}$ \\
\bottomrule
\end{tabularx}
\end{table*}

In this section, we present some applications in the graph mining domain in which our data-dependent sensitivities provide better asymptotic regret bounds compared to the worst case global sensitivity regret bounds, under the assumption that the input graph $G$ is randomly drawn from the Erd\H{o}s--R\'enyi model $G(n, p)$.
We note that the privacy is guaranteed regardless of this assumption, as it only applies to our  analysis of the regret bounds.
Our analysis shows how our methods improve the direct application of EM using global sensitivity in these problems, and do not claim that the methods are superior than state-of-the-art algorithms designed to solve these problems specifically if they exist.

\textbf{Problems.} We are solving several variants of the classic vertex/$k$-set selection.
An $\ell$-clique is a complete graph of size $\ell$.
Informally, the clique-selection problem asks to find a vertex (or an induced subgraph of size $k$) that is included in (or contains, in case of $k$-set) the largest number of $l$-cliques~\cite{everett1998analyzing, konar2022triangle, 10.1093/comnet/cnab052}.
This problem has applications in many domains, for example in biomedical~\cite{10.3389/fphys.2012.00357}, document summarization~\cite{konar2022triangle}, and communication~\cite{10.1145/1557019.1557054}.
Similar problems that estimate the counts of cliques have been studied under DP~\cite{ZhangEtAl2015, KarwaEtAl2011,nguyen2026scalable}, some of which involve calculating maximum cliques per pair or set of vertices, but do not directly output vertex/set selection~\cite{{nguyen2026scalable}}.
\cite{khayatian2026differentiallyprivatescalableestimation} studies a similar exact cardinality $k$-vertex set selection, but with a different objective, namely one with highest density, not highest count of $\ell$-cliques with $\ell \ge 3$ like ours.
We formally define the problem as follows:

\begin{definition}[Clique-selection problems]
For a graph $G$ on $V=[n]$, $n\ge \ell \ge 3$:
\begin{enumerate}[leftmargin=*]
\item In \emph{$K_\ell$-participation vertex selection}, candidates are
$v\in V$ and
$q_\ell(G,v)$ is the number of $\ell$-cliques containing $v$.
\item In \emph{$K_\ell$-densest $k$-set selection}, where
$n\ge k\ge\ell$, candidates are $S\in\binom Vk$ and
$q_{\ell,k}(G,S)$ is the number of $\ell$-cliques in $G[S]$.
\end{enumerate}
\end{definition}

The Clique-selection problem is a special case of a more general family of problems that apply to connected motifs $H$ of size $h$ (subgraphs), such as paths, cycles, stars, instead of $\ell$-cliques.
We define the connected motifs selection in Appendix~\ref{app:apps}.

We assume the edge-differential privacy model (edge-DP)\cite{KarwaEtAl2011,NRS2007}, in which two graphs are neighbors $G\sim G'$ if they share the same set of vertices $V(G) = V(G')$ and differ by one edge $e: E(G) = E(G') \pm e$.

\textbf{Summary of techniques.}
We apply the one-level shifted Laplace method (Cor.~\ref{cor:one-level}) to provide a private high probability upper bound on the local sensitivity of $q_\ell$ and $q_{\ell,k}$.
The worst case $\GS$ can be derived from removing an edge from a complete graph~\cite{ZhangEtAl2015}, i.e., $\GS_{q_\ell} = \binom{n-2}{\ell-2}$ and $\GS_{q_{\ell, k}} = \binom{k-2}{\ell-2}$.
Observe that, when removing an edge $(u,v)$ from an $\ell$-clique, the remaining $\ell-2$ vertices form exactly one $\ell-2$-clique, and all of them are adjacent to both $u$ and $v$.
Hence, the number of $\ell-2$-cliques in the common neighbors of $u$ and $v$ makes up the number of $\ell$-cliques affected by flipping $(u,v)$.
Taking the maximum count over all pairs $u, v\in V(G)$ gives us $\LS_{q_\ell}(G)$.
Next, we analyze the $\GS$ of $\LS_{q_\ell}(G)$, for which we only need to analyze deletions, due to the symmetry.
There are two cases.
First, deleting an edge disconnects $u$ (node $u$ of the pair $u,v$ above) from one of its neighbors not including $v$, which removes at most $\binom{n-3}{\ell-3}$ $\ell-2$-cliques from its common neighbors with $v$.
Second, deleting an edge in the common neighbors of $u, v$ removes at most $\binom{n-4}{\ell-4}$ $\ell-2$-cliques\footnote{We use the convention that $\binom{a}{b} = 0$ if $b <0$ or $b > a$}.
Taking the maximum over both cases gives us the $\GS$ of $\LS_{q_\ell}$ is at most $\binom{n-3}{\ell-3}$.
A similar argument allows us to find $LS_{q_{\ell,k}}(G)$ and its global sensitivity.

\textbf{Regret bounds.}
Let $D$ denote an upper bound on maximum degree in $G(n,p)$ within a sufficiently large constant factor.
Fix $\varepsilon>0$, edge density $p\in[0,1]$, Laplace tail probability and private scale $B$ tail probability $0<\eta,\xi<1/2$, (Cor.~\ref{cor:one-level}), and EM regret tail probability and maximum degree bound tail probability
$0<\rho,\zeta<1$ with $\zeta+\xi+\rho<1$.
For a candidate set of size $m$, define
\begin{equation}
 K_m=\frac{2}{\varepsilon}\log\frac{m}{\rho},
 \quad
 \Lambda=\log\frac{1}{4\eta\xi},
 \quad
  D=1+np+\log\frac{n}{\zeta}.
\end{equation}

EM using the full privacy budget $\varepsilon$
and the global sensitivity has the worst-case guarantee with probability at least $1-\rho$:
\begin{equation}
 \operatorname{Reg}_q(G,R)\le K_m\GS_q
\end{equation}

Table~\ref{tab:app-graph-summary} shows the asymptotic regret bounds of our methods compared to the worst case global sensitivity, for input graphs generated by $G(n,p)$ on four specific problems of the family.
Take the gap, which is the ratio of our regret bound to global sensitivity regret bound, for $K_\ell$ participation vertex for example.
Choosing $\eta,\xi,\zeta= 1/\operatorname{poly}(n)$  makes $\Lambda = O(\log n)$, the gap $O_\ell\!\bigl((D/n)^{\ell-2}+\Lambda/(\varepsilon n)\bigr)$ is $O_\ell((D/n)^{\ell-2} + \log{n}/n)$ when $\eps = \Theta(1)$.
The gaps show significant improvements when $D = o(n)$, or equivalently, $p = o(1)$ for participation vertex problems.
For $k$-set problems, $D = o(k)$ when $\Lambda/(\varepsilon k)=o(1)$ is sufficient for improvements.

\section{Conclusions}

We study the general problem of using local information for private selection, aiming to replace the public worst case global sensitivity in the exponential selection by local variants of sensitivity.
We prove that straightforward, direct replacements using local and smooth sensitivity with fixed smoothness do not guarantee privacy, even under the approximate DP.
Instead, we show that privacy of exponential selection can be achieved by several careful adaptations of local sensitivity, in particular, a privatized one-sided high probability upper bound on local sensitivity, the adaptive PTR framework, and the co-transformation of the smooth sensitivity and utility score, among others. 
We also provide their regret bounds, which are instance-specific and utilize local geometric properties of the private input datasets.
Applying our technique to some vertex/subgraph selection problems shows significant improvements over the global sensitivity for a family of Erd\H{o}s--R\'enyi graphs with reasonable edge density assumptions.

\medskip
\textbf{Acknowledgment.} We have used generative AI tools to assist the development of this work: developing and checking mathematical arguments, preparing and editing the manuscript.
The authors have independently rederived all mathematical results and worked through the proofs in detail, and take full responsibility for the paper’s content.

\newpage

\bibliographystyle{plainurl}
\bibliography{refs}

\appendix
\onecolumn

\section{Proof for Section~\ref{sec:prelims}}

\lemedgeem*

\begin{proof}
Put $c=\varepsilon_s/(2b)$.  Each unnormalized weight changes by at most a
factor $\e^{cb}$.  Summing the reversed pointwise inequality gives
$Z_{x'}(b)/Z_x(b)\le\e^{cb}$.  Therefore
\[
 \frac{Q_x^b(r)}{Q_{x'}^b(r)}
 =\exp\{c(q(x,r)-q(x',r))\}\frac{Z_{x'}(b)}{Z_x(b)}
 \le\e^{2cb}=\e^{\varepsilon_s}.
\]
Exchange $x$ and $x'$ for the other direction.
\end{proof}






\section{Proof and further discussions of Section~\ref{sec:adaptive-ptr}}
\label{app:ptr}

PTR is known for EM~\cite{RedbergEtAl2023}.
For completeness of our discussion, we present here some discussion of PTR for EM with a fixed bound.

Fix a public $b>0$ and let

\[
 \mathcal H_b=\{z:\LS_q(z)>b\},\qquad D_b(x)=d(x,\mathcal H_b).
\]
When finite, $D_b$ is 1-Lipschitz.  A standard PTR mechanism releases a noisy
version of $D_b(x)$, compares it with a public threshold chosen so that the
pass probability is at most $\delta$ whenever $D_b(x)=0$, and runs $Q_x^b$
only on a pass; otherwise it outputs $\bot$.  The noisy test is
$\varepsilon_t$-DP, and every unsafe edge contributes only its false-pass
probability.  Combining the test with Lemma~\ref{lem:edge-em} gives
$(\varepsilon_t+\varepsilon_s,\delta)$-DP.  Conditional on passing, the
ordinary regret bound is
$2b\log(m/\rho)/\varepsilon_s$.

This result is a direct use of PTR~\cite{DworkLei2009} and is not new.
Exact distance is also not essential: a finite nonnegative 1-Lipschitz lower
certificate can replace it.  
Fixed PTR can have a lower test overhead than adaptive search
when a good $b$ is known, but it may abstain and refuse to answer when $b$ is badly chosen.


\section{Proof and further discussions of Section~\ref{sec:smooth}}
\label{app:smooth}

\thmEnvelope*

\begin{proof}
Fix $x\sim x'$ and $r$.  Abbreviate the two breakpoint interpolants by
$B,B'$, and put
$d_0=\min\{\gamma_x(0,r),\gamma_{x'}(0,r)\}$.  Since both zeroth values
dominate their endpoint element sensitivities,
\begin{equation}\label{eq:appendix-d0}
 |u(x,r)-u(x',r)|\le d_0.
\end{equation}
We first show, for every integer $k$,
\begin{equation}\label{eq:appendix-shift}
 B(k)-d_0\ge B'(k-1).
\end{equation}
For $k\ge1$, use $d_0\le\gamma_x(0,r)$ and the envelope shift:
\[
 B(k)-d_0
 \ge\sum_{j=1}^{k-1}\gamma_x(j,r)
 \ge\sum_{j=0}^{k-2}\gamma_{x'}(j,r)=B'(k-1).
\]
For $k=0$, the claim is
$-d_0\ge-\gamma_{x'}(0,r)$.  For $k=-s\le-1$, it is equivalent to
\[
 \sum_{j=0}^{s-1}\gamma_x(j,r)+d_0
 \le\sum_{j=0}^{s}\gamma_{x'}(j,r),
\]
which follows from $d_0\le\gamma_{x'}(0,r)$ and
$\gamma_x(j,r)\le\gamma_{x'}(j+1,r)$.  Affine interpolation of
\eqref{eq:appendix-shift} at consecutive integers gives
\begin{equation}\label{eq:appendix-real-shift}
 B(a)-d_0\ge B'(a-1)\qquad(a\in\mathbb R).
\end{equation}

Let $a=T_x(u(x,r),r)$, so $u(x,r)=B(a)$.  Equations
\eqref{eq:appendix-d0} and \eqref{eq:appendix-real-shift} imply
\[
 u(x',r)\ge u(x,r)-d_0\ge B'(a-1).
\]
Applying the increasing inverse gives
$T_{x'}(u(x',r),r)\ge a-1$.  Exchange $x,x'$ to obtain the other side, so
$|\widetilde u(x,r)-\widetilde u(x',r)|\le1$.  Ordinary EM at global
sensitivity one proves pure $\varepsilon$-DP.
\end{proof}







\section{Proof and further discussions of Section~\ref{sec:apps}}
\label{app:apps}

Fix $\varepsilon>0$, $p\in[0,1]$, $0<\eta,\xi<1/2$, and
$0<\rho,\zeta<1$ with $\zeta+\xi+\rho<1$.
For a candidate set of size $m$, define
\begin{equation}\label{eq:app-common-parameters}
 K_m=\frac{2}{\varepsilon}\log\frac{m}{\rho},
 \qquad
 \Lambda=\log\frac{1}{4\eta\xi}.
\end{equation}
Ordinary exponential selection using the full privacy budget $\varepsilon$
and the global sensitivity has the worst-case guarantee
\begin{equation}\label{eq:app-gs-regret}
 \operatorname{Reg}_q(G,R)\le K_m\GS_q
 \quad\text{with probability at least }1-\rho.
\end{equation}

Our algorithms split the same budget as
$\varepsilon_t=\varepsilon_s=\varepsilon/2$.  Given an exact local scale
$L(G)=\LS_q(G)$ and a public $C\ge\GS_L$, they draw $Z\sim\Lap(1)$, set
\begin{equation}\label{eq:app-private-scale}
 B=\max\left\{1,
 L(G)+\frac{2C}{\varepsilon}
 \left(Z+\log\frac{1}{2\eta}\right)\right\},
\end{equation}
and sample $R$ with probability proportional to
$\exp\{\varepsilon q(G,r)/(4B)\}$.  We use the immediate monotone form of
Corollary~\ref{cor:one-level}: its proof is unchanged when the exact value
$\GS_L$ is replaced by any public upper bound $C\ge\GS_L$.  Thus this
selector is $(\varepsilon,\eta)$-DP.  Moreover, except with probability
$\xi+\rho$ over its internal randomness,
\begin{equation}\label{eq:app-common-regret}
 \operatorname{Reg}_q(G,R)
 \le 2K_m\max\left\{1,L(G)+\frac{2C}{\varepsilon}\Lambda\right\}.
\end{equation}

\dungnote{The paragraph below is written by me}
The factor $\log\frac{1}{2\eta}$ is added to ~\eqref{eq:app-private-scale} to guarantee that it negates the negative Laplacian noise sampled from $Z$ with probability at least $1-\eta$.
Consequently, $B \ge L(G)$ with probability at least $1-\eta$.
Using the same property of the Laplacian noise, a factor $\log\frac{1}{2\xi}$ appears in \eqref{eq:app-common-regret} to set an upper bound on $B$ with probability at least $1-\xi$.
The tail bound of the Gibbs law with the failure probability $\rho$ shows up in the form of $K_m$.

For the random-graph analysis, put
\begin{equation}\label{eq:app-random-scales}
 D=1+np+\log\frac{n}{\zeta},
 \qquad
 \Gamma=1+np^2+\log\frac{n}{\zeta}.
\end{equation}
Standard Bernstein bounds and a union bound give
$d_{\max}(G)+1=O(D)$ and maximum codegree (number of common neighbors of two vertices) $\kappa(G)=O(\Gamma)$ with
probability at least $1-\zeta$.  Write
$K_n=K_m$ for $m=n$ and $K_{n,k}=K_m$ for $m=\binom nk$.
In Table~\ref{tab:app-graph-summary}, $\ell\ge3$ is a fixed clique size,
$H$ is a fixed connected motif on $h\ge3$ vertices, and $k$ is at least the
corresponding motif size.  The hidden constants depend only on the fixed
motif.  Each instance-specific bound holds separately from the other bounds. 

The probability space is jointly over $G\sim G(n,p)$ and the corresponding algorithm run (from noise and sampling), with probability at
least $1-\zeta-\xi-\rho$.

Fix $\ell\ge3$ and $n\ge\ell$, and use the convention $\binom ab=0$ when
$b<0$ or $b>a$.  The following definition covers both vertex participation
and fixed-size dense-motif selection.

\begin{definition}[Clique-selection problems]\label{def:app-clique-selection}
For a graph $G$ on $V=[n]$:
\begin{enumerate}[leftmargin=*]
\item In \emph{$K_\ell$-participation vertex selection}, the candidates are
$v\in V$ and
$q_\ell(G,v)$ is the number of $\ell$-cliques containing $v$.
\item In \emph{$K_\ell$-densest $k$-set selection}, where
$n\ge k\ge\ell$, the candidates are $S\in\binom Vk$ and
$q_{\ell,k}(G,S)$ is the number of $\ell$-cliques in $G[S]$.
\end{enumerate}
Define the two edge-completion profiles
\begin{align}
 A_\ell(G)
 &=\max_{\{u,v\}\in\binom V2}
 C_{\ell-2}\!\left(G[N_G(u)\cap N_G(v)]\right),
 \label{eq:app-clique-profile-vertex}\\
 A_{\ell,k}(G)
 &=\max_{\{u,v\}\in\binom V2}
 \max_{\substack{T\subseteq N_G(u)\cap N_G(v)\\|T|\le k-2}}
 C_{\ell-2}(G[T]),
 \label{eq:app-clique-profile-k}
\end{align}
where $C_j(F)$ is the number of $j$-cliques in $F$ and $C_1(F)=|V(F)|$, and $N_G(v)$ is the set of neighbor vertices of $v$ in graph $G$.
\end{definition}

\begin{algorithm}
\caption{Private clique selection}
\label{alg:app-clique}
\begin{algorithmic}[1]
\REQUIRE Graph $G$ and one of the two problems in
Definition~\ref{def:app-clique-selection}.
\IF{$K_\ell$-participation vertex selection}
  \STATE Set $\R\gets V$, $q\gets q_\ell$, $L\gets A_\ell$, and
  $C\gets\binom{n-3}{\ell-3}$.
\ELSE
  \STATE Set $\R\gets\binom Vk$, $q\gets q_{\ell,k}$,
  $L\gets A_{\ell,k}$, and $C\gets\binom{k-3}{\ell-3}$.
\ENDIF
\STATE Compute $L(G)$.
\STATE Draw the private scale $B$ according to
\eqref{eq:app-private-scale}.
\STATE Sample $R\in\R$ with weights
$w(r)=\exp\{\varepsilon q(G,r)/(4B)\}$ for $r\in\R$.
\RETURN $R$.
\end{algorithmic}
\end{algorithm}

\begin{theorem}[Clique-selection regret on $G(n,p)$]
\label{thm:app-clique-gnp}
For every input graph with $n\ge\ell$,
Algorithm~\ref{alg:app-clique} is $(\varepsilon,\eta)$-DP.  Let
$R_V$ and $R_k$ denote the outputs of its vertex and $k$-set
instantiations, respectively.  If $G\sim G(n,p)$ and $\ell$ is fixed, then
the following statements each hold with probability at least
$1-\zeta-\xi-\rho$:
\begin{align}
 \operatorname{Reg}_{q_\ell}(G,R_V)
 &=O_\ell\!\left(
 K_n\left[D^{\ell-2}
 +\frac{n^{\ell-3}}{\varepsilon}\Lambda\right]\right),
 \label{eq:app-clique-vertex-regret}\\
 \operatorname{Reg}_{q_{\ell,k}}(G,R_k)
 &=O_\ell\!\left(
 K_{n,k}\left[\min\{k,D\}^{\ell-2}
 +\frac{k^{\ell-3}}{\varepsilon}\Lambda\right]\right).
 \label{eq:app-clique-k-regret}
\end{align}
The corresponding full-budget global-sensitivity bounds are exactly
\begin{equation}\label{eq:app-clique-gs}
 K_n\binom{n-2}{\ell-2}
 \quad\text{and}\quad
 K_{n,k}\binom{k-2}{\ell-2}.
\end{equation}
For $\ell=3$, one may replace $D$ by $\Gamma$ in
\eqref{eq:app-clique-vertex-regret} and inside
$\min\{k,D\}$ in \eqref{eq:app-clique-k-regret}.
\end{theorem}

\begin{proof}
The sensitivity analysis below is inspired by the sensitivity analysis of counting $\ell$-cliques in a graph under edge-DP by~\cite{ZhangEtAl2015}, which also uses the highest number of cliques involving a pair of vertices.
Toggling an edge $uv$ changes exactly the $\ell$-cliques that contain
$uv$.  Their other $\ell-2$ vertices form a clique in
$G[N_G(u)\cap N_G(v)]$.  A candidate endpoint contains every changed
clique, and any feasible set $T$ in
\eqref{eq:app-clique-profile-k} can be completed to a $k$-set, i.e., adding extra vertices to $T$ to get $|T|=k$.  Hence
\[
 \LS_{q_\ell}(G)=A_\ell(G),
 \qquad
 \LS_{q_{\ell,k}}(G)=A_{\ell,k}(G).
\]
The complete graph gives
$\GS_{q_\ell}=\binom{n-2}{\ell-2}$ and
$\GS_{q_{\ell,k}}=\binom{k-2}{\ell-2}$.

A second edge toggle can add or remove at most one vertex from a fixed
common-neighbor graph, or toggle one of its internal edges.  For the
$k$-set profile, delete a newly available vertex from an optimizer after the
toggle.  The remaining set was feasible before the toggle, and at most
$\binom{k-3}{\ell-3}$ of its $(\ell-2)$-cliques are lost.  An internal-edge
toggle changes any fixed feasible set by at most
$\binom{k-4}{\ell-4}\le\binom{k-3}{\ell-3}$.  Deletion is symmetric, and
the same argument with at most $n-2$ common neighbors treats $A_\ell$.
Finally,
$|\max_i f_i-\max_i g_i|\le\max_i|f_i-g_i|$ transports these fixed-pair
bounds through the outer maximum.  Thus
\begin{equation}\label{eq:app-clique-second}
 \GS_{A_\ell}\le\binom{n-3}{\ell-3},
 \qquad
 \GS_{A_{\ell,k}}\le\binom{k-3}{\ell-3}.
\end{equation}
These are the public values $C$ used in the algorithm, so privacy follows
from Corollary~\ref{cor:one-level}.

The deterministic bounds
\[
 A_\ell(G)\le\binom{d_{\max}(G)}{\ell-2},
 \qquad
 A_{\ell,k}(G)
 \le\binom{\min\{d_{\max}(G),k-2\}}{\ell-2}
\]
and $d_{\max}(G)+1=O(D)$ on $G(n,p)$, followed by
\eqref{eq:app-common-regret}, give
\eqref{eq:app-clique-vertex-regret}--
\eqref{eq:app-clique-k-regret}.  When $\ell=3$,
$A_3(G)=\kappa(G)$ and
$A_{3,k}(G)=\min\{\kappa(G),k-2\}$; the codegree event
$\kappa(G)=O(\Gamma)$ gives the sharper claim.
\end{proof}

\textbf{Connected-motif generalization}

In the next part, we generalize the $\ell$-clique vertex/subgraph selection to any connected motif $H$ of size $h$, of which an $h$-clique is a special case of a complete $H$.
We consider the non-induced embedding, i.e., a copy of $H$ in $G$ only preserves the edges of $H$.
It contrasts with the induced embedding that preserves both edges and non-edges.

\begin{definition}[Connected-motif selection problems]
\label{def:app-motif-selection}
Fix a connected simple graph $H=(W,F)$ with
$h=|W|\ge3$, $n\ge h$, $e_H=|F|$, and
$a_H=|\mathrm{Aut}(H)|$.  An $H$-copy is an
injective edge-preserving embedding of $H$ into $G$, modulo its $a_H$
automorphisms.
\begin{enumerate}[leftmargin=*]
\item In \emph{$H$-participation vertex selection}, the candidates are
$v\in V$ and $q_H(G,v)$ is the number of $H$-copies containing $v$.
\item In \emph{$H$-densest $k$-set selection}, where $n\ge k\ge h$, the
candidates are $S\in\binom Vk$ and $Q_{H,k}(G,S)$ is the number of
$H$-copies contained in $G[S]$.
\end{enumerate}
For a host pair $e\in\binom V2$, let $G+e$ add $e$.  Let
$C^H_{v,e}(G)$ count the $H$-copies in $G+e$ that use $e$ and contain $v$,
and let $C^H_{S,e}(G)$ additionally require the copy to lie in $S$.  Define
\begin{equation}\label{eq:app-motif-profiles}
 L_H^V(G)=\max_{v,e}C^H_{v,e}(G),
 \qquad
 L_H^k(G)=\max_{S\in\binom Vk}\max_{e\subseteq S}C^H_{S,e}(G).
\end{equation}
\end{definition}

\begin{algorithm}
\caption{Private connected-motif selection}
\label{alg:app-motif}
\begin{algorithmic}[1]
\REQUIRE Graph $G$ and one of the two problems in
Definition~\ref{def:app-motif-selection}.
\STATE Set $c_H\gets 4e_H^2/a_H$.
\IF{$H$-participation vertex selection}
  \STATE Set $\R\gets V$, $q\gets q_H$, $L\gets L_H^V$, and
  $C\gets c_Hn^{h-3}$.
\ELSE
  \STATE Set $\R\gets\binom Vk$, $q\gets Q_{H,k}$,
  $L\gets L_H^k$, and $C\gets c_Hk^{h-3}$.
\ENDIF
\STATE Compute $L(G)$.
\STATE Draw the private scale $B$ according to
\eqref{eq:app-private-scale}.
\STATE Sample $R\in\R$ with weights
$w(r)=\exp\{\varepsilon q(G,r)/(4B)\}$ for $r\in\R$.
\RETURN $R$.
\end{algorithmic}
\end{algorithm}

\begin{theorem}[Connected-motif regret on $G(n,p)$]
\label{thm:app-motif-gnp}
For every input graph with $n\ge h$,
Algorithm~\ref{alg:app-motif} is $(\varepsilon,\eta)$-DP.  Let
$R_V$ and $R_k$ denote the outputs of its vertex and $k$-set
instantiations, respectively.  If $G\sim G(n,p)$ and $H$ is fixed, then
the following statements each hold with probability at least
$1-\zeta-\xi-\rho$:
\begin{align}
 \operatorname{Reg}_{q_H}(G,R_V)
 &=O_H\!\left(
 K_n\left[D^{h-2}
 +\frac{n^{h-3}}{\varepsilon}\Lambda\right]\right),
 \label{eq:app-motif-vertex-regret}\\
 \operatorname{Reg}_{Q_{H,k}}(G,R_k)
 &=O_H\!\left(
 K_{n,k}\left[\min\{k,D\}^{h-2}
 +\frac{k^{h-3}}{\varepsilon}\Lambda\right]\right).
 \label{eq:app-motif-k-regret}
\end{align}
The corresponding full-budget global-sensitivity bounds are
\begin{equation}\label{eq:app-motif-gs}
 \frac{2e_H}{a_H}(n-2)_{h-2}K_n
 \quad\text{and}\quad
 \frac{2e_H}{a_H}(k-2)_{h-2}K_{n,k},
\end{equation}
where $(z)_r=z(z-1)\cdots(z-r+1)$.
\end{theorem}

\begin{proof}
Toggling $e$ creates or destroys exactly the non-induced $H$-copies that use
$e$; there is no cancellation: adding an edge cannot invalidate any copy of $H$ in the non-induced embedding, in contrast to induced embedding in which adding an edge can mismatch a non-edge of $H$.  Maximizing over the candidate and $e$
therefore gives
$\LS_{q_H}=L_H^V$ and $\LS_{Q_{H,k}}=L_H^k$.  In a complete host, i.e., the original graph $G$, choose
one of the $2e_H$ oriented pattern edges for the endpoints of $e$, inject
the remaining $h-2$ vertices, and divide by $a_H$.
There are $2e_H$ oriented pattern edges because each edge of $H$ has two ways to map to an edge of $e$ of host $G$.
This proves the two
global sensitivities in \eqref{eq:app-motif-gs}.

For a fixed term in \eqref{eq:app-motif-profiles}, toggling its distinguished
edge $e$ leaves $G+e$ unchanged.  Any other second toggle changes only
copies using both distinct host edges.  Assigning them to two oriented
pattern edges takes at most $2e_H(2eH -2) \le 4e_H^2$ choices, and their union fixes at least
three host vertices, as the other case is two disjoint edges that fix four host vertices.  Precomposition by $\mathrm{Aut}(H)$ acts freely on
the relevant embeddings, so its orbits have size $a_H$.  Division by
$a_H$, followed by
$|\max_i f_i-\max_i g_i|\le\max_i|f_i-g_i|$, gives
\begin{equation}\label{eq:app-motif-second}
 \GS_{L_H^V} \le c_H {(n-3)}_{h-3} \le c_Hn^{h-3},
 \qquad
 \GS_{L_H^k}\le c_H{(k-3)}_{h-3} \le c_Hk^{h-3},
\end{equation}

with $c_H$ as defined in Algorithm~\ref{alg:app-motif}.
Thus the values $C$ in Algorithm~\ref{alg:app-motif} are valid, and privacy
again follows from Corollary~\ref{cor:one-level}.

Finally, after fixing the image of a pattern edge, connectedness lets us
grow a spanning tree from its two endpoints.  Each of the remaining $h-2$
vertices has at most $d_{\max}(G)+1$ choices, as $G + e$ has maximum degree $d_{\max}(G)+1$, and at most $k$ choices inside
a $k$-set.  Hence
\[
 L_H^V(G)=O_H((d_{\max}(G)+1)^{h-2}),
 \qquad
 L_H^k(G)=O_H(\min\{k,d_{\max}(G)+1\}^{h-2}).
\]
Combining the event $d_{\max}(G)+1=O(D)$ with
\eqref{eq:app-common-regret} proves
\eqref{eq:app-motif-vertex-regret}--\eqref{eq:app-motif-k-regret}.
\end{proof}

\end{document}